\documentclass[preprint,12pt]{./elsarticle}

\usepackage{lineno,hyperref}

\usepackage{graphicx}
\usepackage{subcaption}
\usepackage{amssymb}
\usepackage{amsfonts}
\usepackage{amsmath} 
\usepackage{amsbsy} 
\usepackage{amsthm} 
\usepackage{mathrsfs} 
\usepackage{epstopdf} 
\usepackage{dsfont} 
\usepackage{pythonhighlight}
\usepackage{float} 
\usepackage{multirow}
\usepackage{booktabs}
\usepackage{hhline}

\newtheorem{definition}{Definition}
\newtheorem{Theorem}{Theorem}
\newtheorem{hypothesis}{Hypothesis}
\newtheorem{corollary}{Corollary}

\begin{document}

\begin{frontmatter}

\title{Exponential Sigmoid Equation for Modelling Cell Growth in a Confined Space, Log-Normal Distribution for Modelling Cell Area Distribution of Dense Colonies and Other Methods\tnoteref{t1}}
\tnotetext[t1]{This work is based on the data provided by Wheeler Bio, Inc.}

\author[mymainaddress]{Kavinda Jayawardana\corref{mycorrespondingauthor} and Brad Turner \corref{mycorrespondingauthor}}
\address[mymainaddress]{TEK Optima Research Ltd, Unit 10 Westcroft Buiness Park Oakdene Drive, Three Legged Cross, Wimborne BH21 6FQ}
\ead{kavjayawardana@tekoptimaresearch.com; bradturner@tekoptimaresearch.com}
\cortext[mycorrespondingauthor]{Corresponding authors}

\begin{abstract}
Based on the growth patterns of 166 CHO monoclones observed over a 15 day period, we show that the standard population growth in a confined space equation, i.e. the sigmoid/logistic function, is alone does not capture the complex behaviour of the cell growth in a confined space. Thus, combining the sigmoid function and the exponential of the sigmoid function, we present a more accurate model for modelling cell growth in a confined space. We also present a working algorithm to obtain population growth variables (growth capacity, growth time and growth rate), model the growth patterns of the CHO monoclones, and we include subset of the dataset, along with a sample python script for the reader to replicate the results. Furthermore, we derive a model for cell confluence growth in a confined space, numerically model the confluence and present the reader with a working algorithm. With Kolmogorov-Smirnov analysis conducted on the area of the CHO monoclones, we show that the cell area  of the incipient population is normally distributed, the sparse cell population is gamma distributed and the dense colony population is log-normally distributed. Thus, we further derive models for the mean, the standard deviation, the coefficient of variation and the inverse coefficient of variation for the log cell area growth in a confined space, numerically model them and present the reader with working algorithms. Finally, based on the growth patterns of another 48 CHO monoclones observed over a 16 day period, and their titer and viability measurements, we find the correlation coefficients with our calculated growth variables, and titer and viability measurements, and show that our derived growth variables can be used to predict the productivity and the health of a cell. Thus, we conclude our study by demonstrating that the productivity and the health of a cell (also the overall population) are interdependent.
\end{abstract}


\end{frontmatter}


\section{Introduction}

\begin{figure*}[h!]
\centering

\begin{minipage}[t]{0.43\textwidth}
    \centering

    \begin{subfigure}{0.31\textwidth}
        \includegraphics[height=\textwidth]{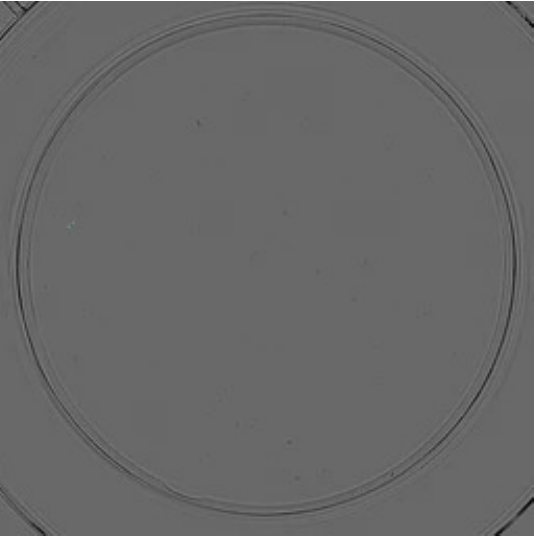}
        \caption{Day 1.}
    \end{subfigure}
    \hfill
    \begin{subfigure}{0.31\textwidth}
        \includegraphics[height=\textwidth]{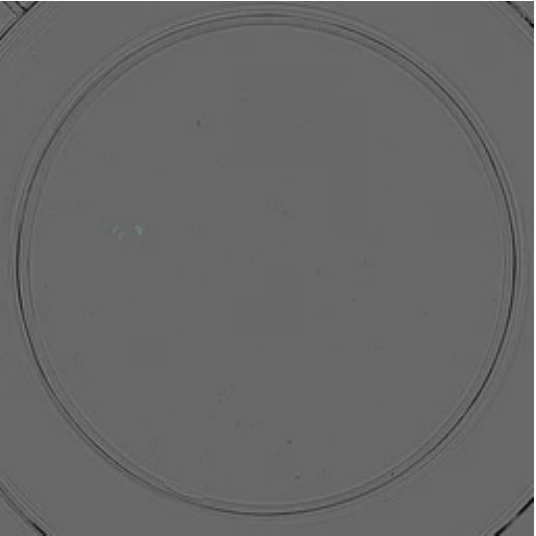}
        \caption{Day 2.}
    \end{subfigure}
    \hfill
    \begin{subfigure}{0.31\textwidth}
        \includegraphics[height=\textwidth]{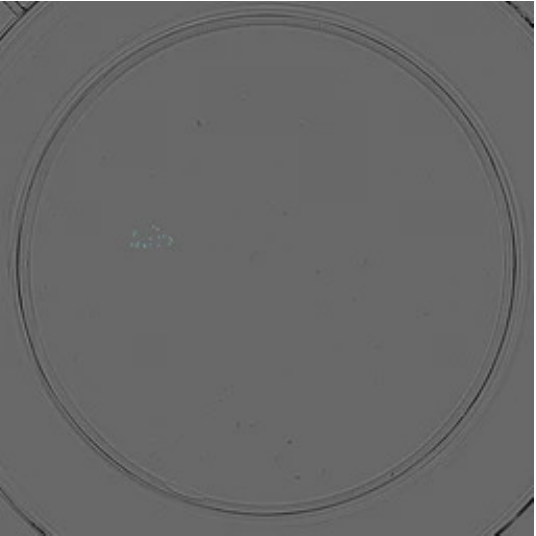}
        \caption{Day 3.}
    \end{subfigure}

    \vspace{0.5em}

    \begin{subfigure}{0.31\textwidth}
        \includegraphics[height=\textwidth]{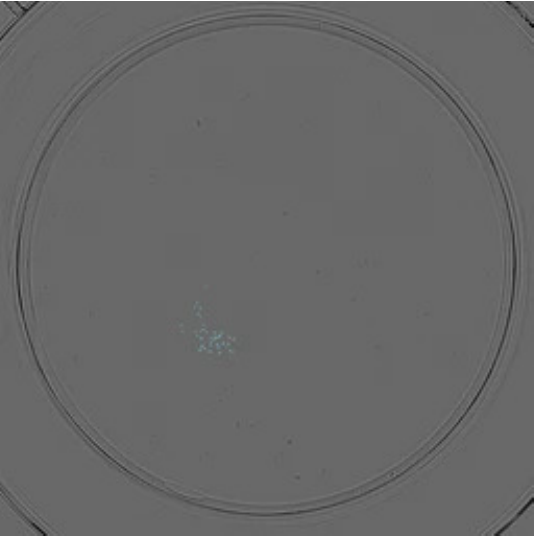}
        \caption{Day 4.}
    \end{subfigure}
    \hfill
    \begin{subfigure}{0.31\textwidth}
        \includegraphics[height=\textwidth]{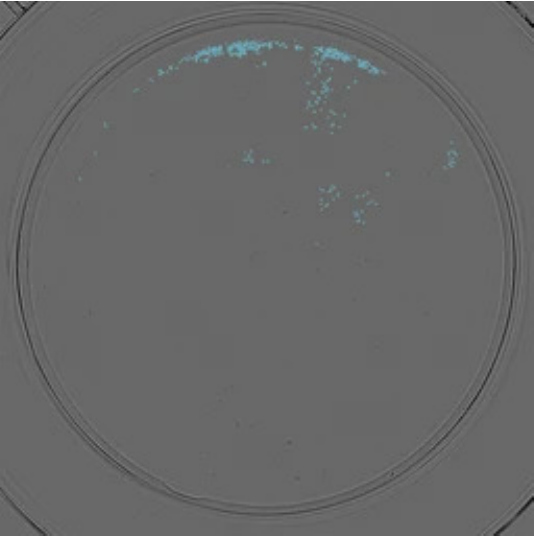}
        \caption{Day 7.}
    \end{subfigure}
    \hfill
    \begin{subfigure}{0.31\textwidth}
        \includegraphics[height=\textwidth]{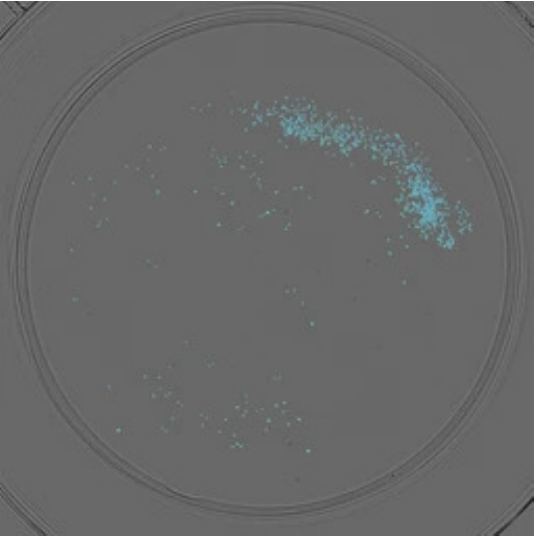}
        \caption{Day 8.}
    \end{subfigure}

    \vspace{0.5em}

    \begin{subfigure}{0.31\textwidth}
        \includegraphics[height=\textwidth]{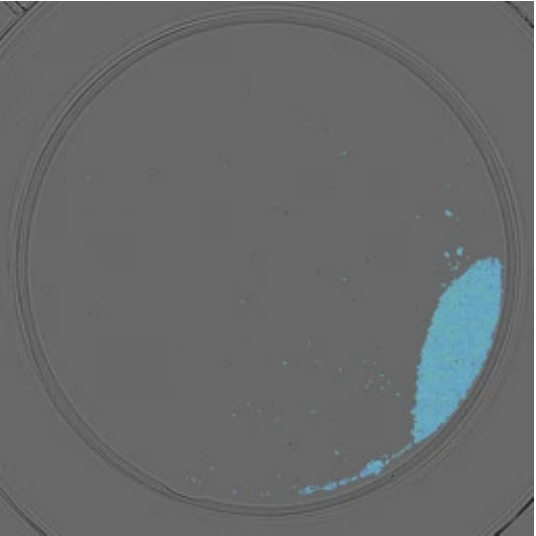}
        \caption{Day 10.}
    \end{subfigure}
    \hfill
    \begin{subfigure}{0.31\textwidth}
        \includegraphics[height=\textwidth]{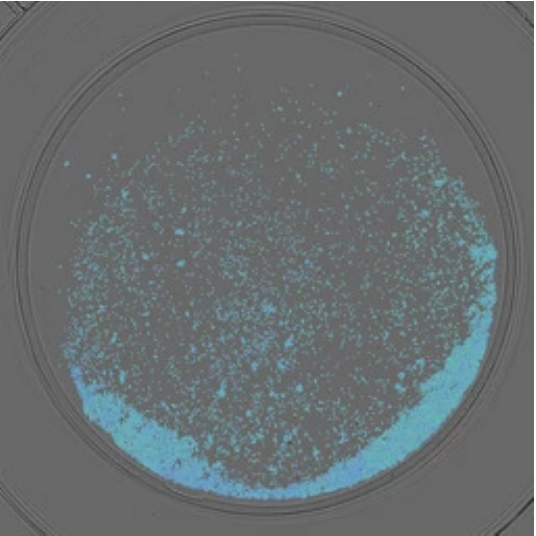}
        \caption{Day 11.}
    \end{subfigure}
    \hfill
    \begin{subfigure}{0.31\textwidth}
        \includegraphics[height=\textwidth]{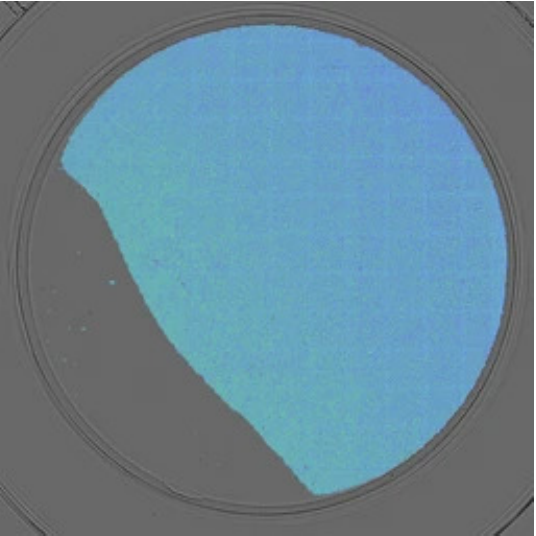}
        \caption{Day 14.}
    \end{subfigure}

\end{minipage}
\hfill
\begin{minipage}[t]{0.55\textwidth}
    \centering
	\vspace{-6.25em}
    \begin{subfigure}{\linewidth}
        \includegraphics[height=\linewidth]{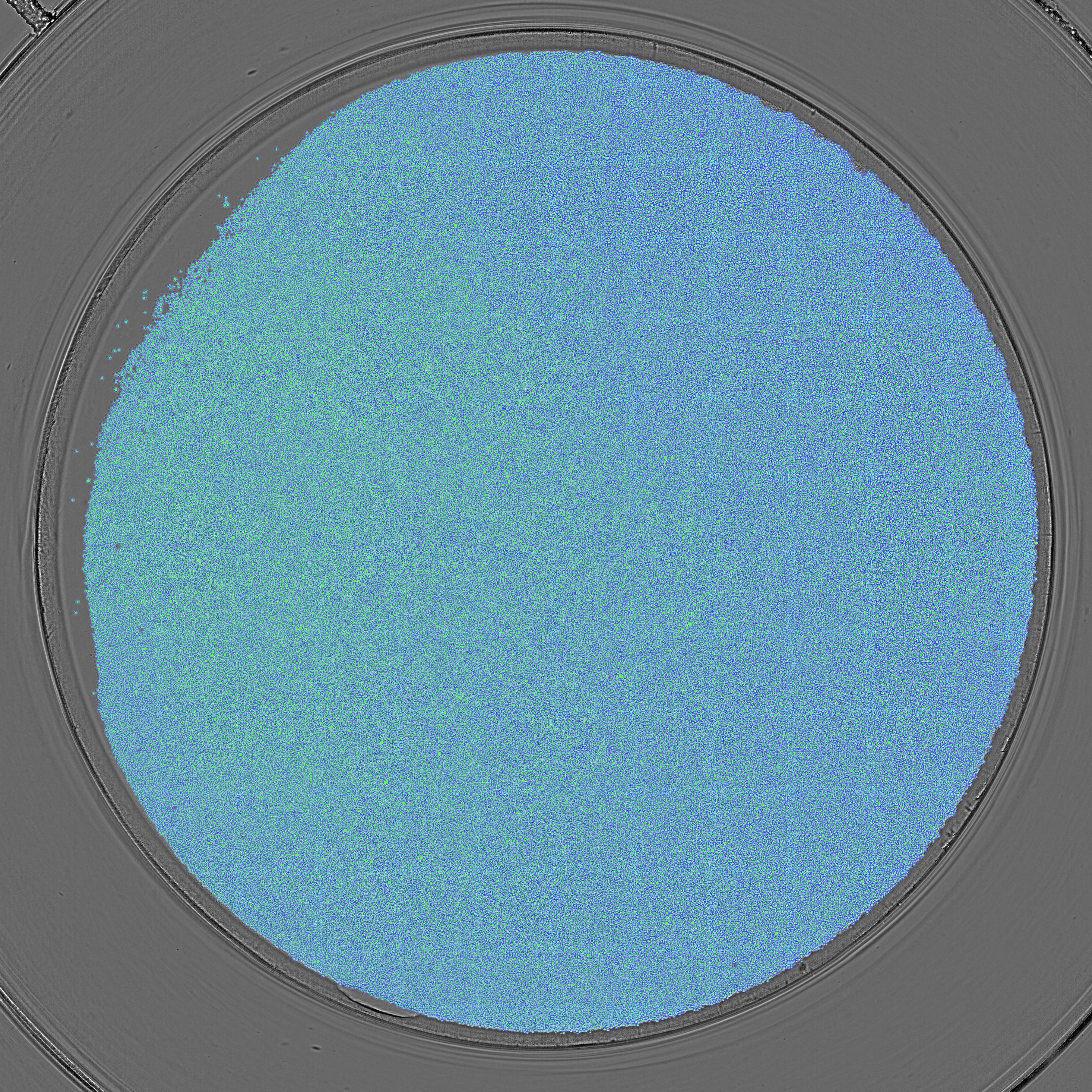}
        \caption{Day 15.}
    \end{subfigure}
\end{minipage}
\vspace*{-0.5em}
\caption{Growth of a CHO monoclone ($cho\#11603$) over a 15 day period.}
\label{cho-15}
\end{figure*}

Cell counting is a key part of biotechnology, life sciences, diagnostics and biopharmaceuticals industries, as it plays a vital role in areas such as drug discovery, drug development, stem cell research, cell line development, cancer research and complete blood count analysis. Grand View Research  organisation estimate that the global market capitalisation of the cell counting industry is \$9.48 billion in 2024, and the market is projected to grow to estimated \$15.46 billion by 2030 \cite{grandviewresearch}.\\

A crucial part in the field of cell counting is detecting monoclonality, i.e. a population that originates from a single progenitor, as it has major applications in the field of monoclonal antibody therapy, where treatments include targeted cancer therapies, treating infectious diseases, and treating autoimmune and neurological disorders \cite{yao2013advances}. Thus, there are a huge array of cell line development companies and products that are dedicated detecting monoclonality, which include the companies  Nova Biomedical \cite{Solentim} and Wheeler Bio \cite{wheelerbio}. Key areas of interest in this field is the modelling of cell growth in a confined space, and where one of the key modelling techniques is the use of the logistic function to model cell growth in a  confined space (more generally, an organism's population growth in an environment with limited resources) \cite{Biology, SINGH2021257}. Another notable equation for modelling a population in a confined space is the Gompertz function, which is a generalised case of the logistic function \cite{zwietering1990modeling}. \\

At the courtesy of Wheeler Bio, Inc. \cite{wheelerbio}, we obtain a dataset of 166 CHO monoclones observed over a 15 day period, where  the cells seeded with Solentim VIPS \cite{VIPS} to ensure high probability of monoclonality, images of the wells are scanned with the Solentim Cell Metric \cite{cellmetric} and the cells count is calculated  with TEK Optima Research Ltd. DeepInsight\textsuperscript{\textregistered} cell analysis software \cite{tekor}. See figure \ref{cho-15} for a growth of a CHO monoclone over a 15 day period, where cell centres (green) and cell borders (blue) are highlighted by the DeepInsight\textsuperscript{\textregistered} cell analysis software \cite{tekor}. With data analysis, we find that the models in the literature that describe the behaviour of cell growth in confined spaces do not adequately describe the complex growth patterns that we observe. Yin \emph{et al.}'s  beta growth function \cite{yin2003flexible} comes remarkably close to modelling growth patterns that we observe (even though the it is used in modelling the growth of plants); however, it still cannot fully model the growth speeds and growth accelerations that we observe in the dataset. Thus, motivating us to pursuit a more accurate model by combining the sigmoid function and the exponential of the sigmoid function. Note that sigmoid function is a specific case of the logistic function \cite{dr2025machine}. However, herein, we refer to the logistic function as the sigmoid function, as we modify the sigmoid function to derive our model.\\

Observing cell size over several growth cycles, Jia \emph{et al.} \cite{jiacell} show that growing populations usually have a right-skewed distribution due to the birth of new daughter cells in the population. The authors also show that the distribution may vary due to a fast increase for small cells, followed by a slow decay for moderately large cells and a fast decay for exceptionally large cells; however, the cell size distribution can still be modelled by the gamma distribution. Efficacy of modelling cell size distribution with the gamma distribution is also demonstrated by Golubev \cite{golubev2016applications}. Lenz \emph{et al.} \cite{lenz2016estimating} demonstrate that cell sizes are in fact log-normally distributed, at least in tissue samples, as sectioning a tissue at random depths leads to an artificially skewed distribution of smaller measured cells.  However, we find no convincing evidence in the literature where the authors perform any statistical analysis (with Kolmogorov-Smirnov test or otherwise) to perform goodness of fit for the aforementioned distributions (or any distributions). We did find Demidenko \cite{demidenko} applying Kolmogorov-Smirnov test to histology images from untreated and treated breast cancer tumours as a method of comparing images; however, the author is not estimating cell size distributions. Thus, analysing change in cell size distribution as the population grows and applying rigours statistical analysis will also be a subject of investigation in our work.\\

Marshall \cite{marshall2012organelle} propose that as a result of molecular mechanisms that regulate internal structure dimensions are proportional to the cell size, a growth rate of a cell is proportional to its current size, and larger cells likely to synthesise more proteins and grow faster. Based on growth patterns of further 48 CHO monoclones observed over a 16 day period, and their titer and viability measurements, we also investigate how the growth patterns (both population, cell size and other metrics) are correlated with the productivity and the health of the cell, where titer is a measurement of the concentration of a substance in a solution, and viability is the measurement of the live to dead cell ratio in a cell culture.

\section{Cell Count \label{cell-count}}

\begin{figure}[h!]
     \centering
     \begin{subfigure}[b]{0.49\textwidth}
         \centering
         \includegraphics[height=\textwidth]{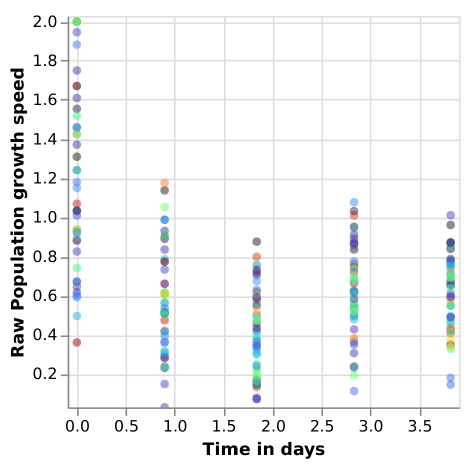}
         \caption{Initial stage growth speed.}
         \label{growth-speed0}
     \end{subfigure}
     \hfill
     \begin{subfigure}[b]{0.49\textwidth}
         \centering
         \includegraphics[height=\textwidth]{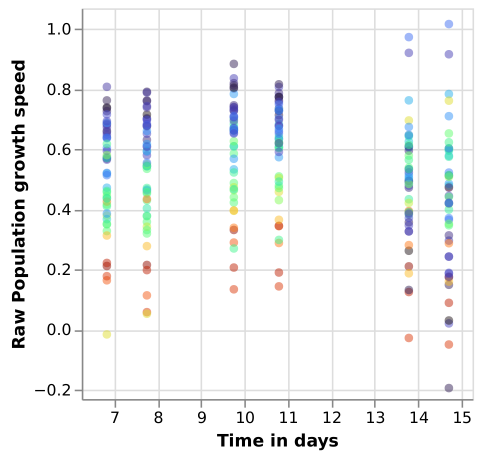}
         \caption{Later stage growth speed.}
         \label{growth-speed1}
     \end{subfigure}
     \hfill
     \begin{subfigure}[b]{0.49\textwidth}
         \centering
         \includegraphics[height=\textwidth]{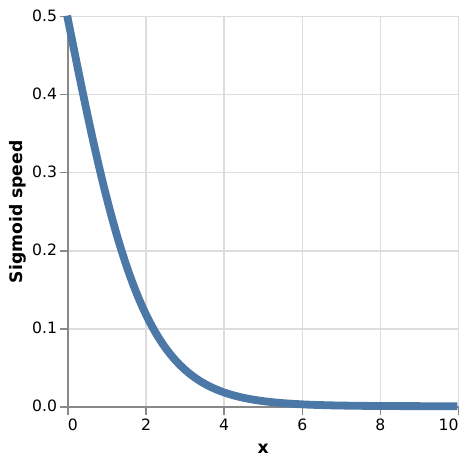}
         \caption{Sigmoid growth speed.}
         \label{growth-speed-sig0}
     \end{subfigure}
     \hfill
     \begin{subfigure}[b]{0.49\textwidth}
         \centering
         \includegraphics[height=\textwidth]{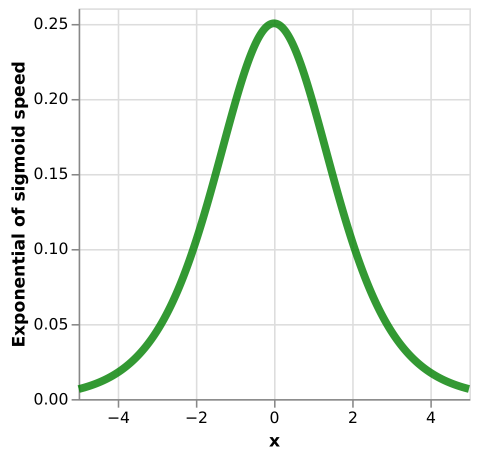}
         \caption{Exp-sigmoid growth speed.}
         \label{growth-speed-sig1}
     \end{subfigure}
        \caption{Growth speed for initial stage and latter stage of cell growth.}
        \label{growth-speed}
\end{figure}

\begin{figure}[h!]
     \centering
     \begin{subfigure}[b]{0.49\textwidth}
         \centering
         \includegraphics[height=\textwidth]{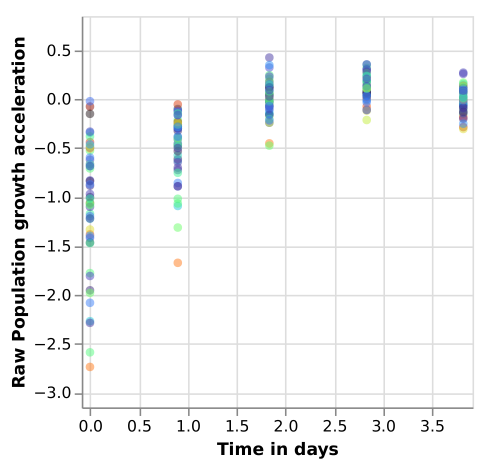}
         \caption{Initial stage growth acceleration.}
         \label{growth-acceleration0}
     \end{subfigure}
     \hfill
     \begin{subfigure}[b]{0.49\textwidth}
         \centering
         \includegraphics[height=\textwidth]{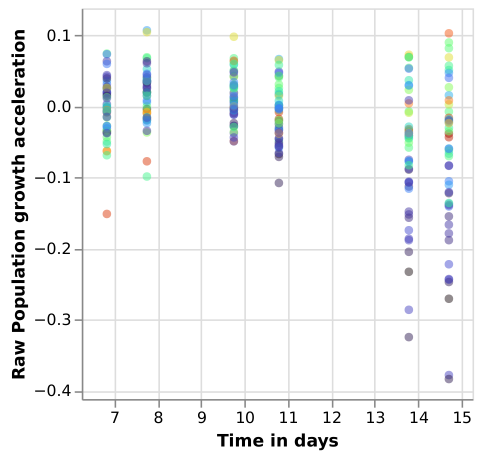}
         \caption{Later stage growth acceleration.}
         \label{growth-acceleration1}
     \end{subfigure}
     \hfill
     \begin{subfigure}[b]{0.49\textwidth}
         \centering
         \includegraphics[height=\textwidth]{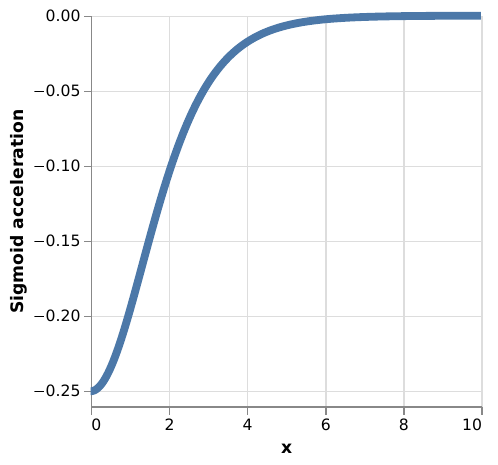}
         \caption{Sigmoid growth acceleration.}
         \label{growth-acceleration-sig0}
     \end{subfigure}
     \hfill
     \begin{subfigure}[b]{0.49\textwidth}
         \centering
         \includegraphics[height=\textwidth]{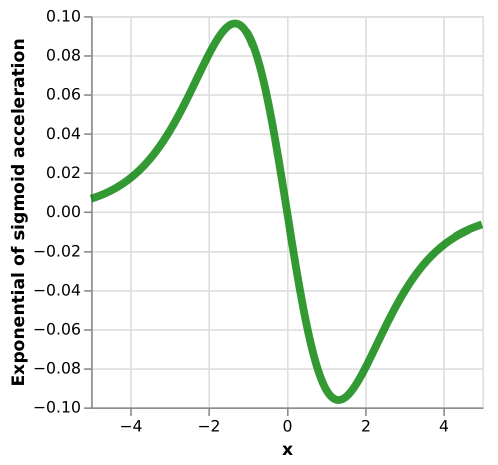}
         \caption{Exp-sigmoid growth acceleration.}
         \label{growth-acceleration-sig1}
     \end{subfigure}
        \caption{Growth acceleration for initial stage and later stage of cell growth.}
        \label{growth-acceleration}
\end{figure}

Consider figure \ref{growth-speed}, where it represents the growth speed of 166 CHO monoclones in 96-well microplates observed over a 15 day period, the set \emph{CHO2023}. Assuming cell growth can be modelled by the sigmoid function, i.e. $ \log(y_n(x)) = \log(\boldsymbol{\sigma}(x))$ where $y_n(x)$ is the number of cells (hence the subscript  $n$), $x$ is the time and where the \emph{sigmoid} function defined as as follows,
\begin{align}\label{sigmoid}
\boldsymbol{\sigma}(x) = \frac{1}{1 + \exp(-x)},
\end{align}
(where this is the standard model of cell growth in confined space), and assuming cell growth can also be modelled by exponential of the sigmoid equation, i.e. $ \log(y_n(x)) = \boldsymbol{\sigma}(x)$, and by examining sub-figures \ref{growth-speed0} with \ref{growth-speed-sig0} and sub-figures \ref{growth-speed1} with \ref{growth-speed-sig1}, we can see that sigmoid function is better at modelling the initial growth speed of the cells and the exponential of the sigmoid function is better at modelling later growth speed of the cells. Note that by \emph{growth speed}, we mean the first-order central difference of the log cell count with respect to time, i.e. $\frac{d\log(y_n(x))}{dx}|_i = \frac{\log(y_n)_{i+1} - \log(y_n)_{i-1}}{x_{i+1} - x_{i-1}}$, where $i$ is the $i$-th sample in the dataset. Also, the cell growth speed is calculated by \emph{NumPy} \emph{gradient} function with \emph{edge\_order}$=2$ \cite{numpy.gradient}, the speed of sigmoid function calculated as $\frac{d\log(y_n(x))}{dx} = 1 - \boldsymbol{\sigma}(x)$, the speed of the exponential of the sigmoid function is calculated as $\frac{d\log(y_n(x))}{dx} = \boldsymbol{\sigma}(x) (1 - \boldsymbol{\sigma}(x))$, by \emph{raw}, we mean any data point that is not fitted to a model and by \emph{log}, we mean the natural-log.\\

Should we examine the growth acceleration, expressed in figure \ref{growth-acceleration}, we can see that that sigmoid function is still better at modelling the initial growth acceleration of the cells (see sub-figures \ref{growth-acceleration0} and \ref{growth-acceleration-sig0}) and the exponential of the sigmoid function is still better at modelling later growth acceleration of the cells (see sub-figures \ref{growth-acceleration1} and \ref{growth-acceleration-sig1}). Note that by \emph{growth acceleration}, we mean the second-order central difference of the log cell count with respect to time, i.e. $\frac{d^2\log(y_n(x))}{dx^2} |_i = 4 \frac{\log(y_n)_{i+1} - 2 \log(y_n)_{i} + \log(y_n)_{i-1}}{(x_{i+1} - x_{i-1})^2}$. Also, the cell growth acceleration is again calculated by \emph{NumPy} \emph{gradient} function with \emph{edge\_order}$=2$ \cite{numpy.gradient}, the acceleration of sigmoid function calculated as $\frac{d^2\log(y_n(x))}{dx^2} = -  \boldsymbol{\sigma}(x) (1 - \boldsymbol{\sigma}(x))$ and the acceleration of the exponential of the sigmoid function is calculated as $\frac{d^2\log(y_n(x))}{dx^2} = \boldsymbol{\sigma}(x) (1 - \boldsymbol{\sigma}(x)) (1 - 2\boldsymbol{\sigma}(x))$.\\

Prior to any further analysis, first we define equations related to the sigmoid equations as follows,
\begin{definition} \label{sigmoid-diff}
Given that $\boldsymbol{\sigma}(\theta (x - \gamma))$ is the sigmoid function that is translated in the $x$-axis in by a factor of $\gamma$ and scaled in the $x$-axis by a factor of $\frac{1}{\theta}$, then it and its first four normalised derivatives can be expressed as follows,
\begin{align*}
\boldsymbol{\sigma}(\theta (x - \gamma)) = ~&\frac{1}{1 + \exp(-\theta(x - \gamma))},\\
\frac{1}{\theta} \boldsymbol{\sigma}^{(1)}(\theta (x - \gamma))  =~ &\boldsymbol{\sigma}(\theta (x - \gamma)) (1- \boldsymbol{\sigma}(\theta (x - \gamma))),\\
\frac{1}{\theta^2} \boldsymbol{\sigma}^{(2)}(\theta (x - \gamma))  = ~ & \boldsymbol{\sigma}(\theta (x - \gamma)) (1- \boldsymbol{\sigma}(\theta (x - \gamma)))\big[1-2 \boldsymbol{\sigma}(\theta (x - \gamma))],\\
\frac{1}{\theta^3} \boldsymbol{\sigma}^{(3)}(\theta (x - \gamma))  = ~ & \boldsymbol{\sigma}(\theta (x - \gamma)) (1- \boldsymbol{\sigma}(\theta (x - \gamma)))\\
&\big[1-6 \boldsymbol{\sigma}(\theta (x - \gamma)) + 6 \boldsymbol{\sigma}(\theta (x - \gamma))^2], \text{and}\\
\frac{1}{\theta^4} \boldsymbol{\sigma}^{(4)}(\theta (x - \gamma))  = ~ & \boldsymbol{\sigma}(\theta (x - \gamma)) (1- \boldsymbol{\sigma}(\theta (x - \gamma)))\\
&\big[1-14 \boldsymbol{\sigma}(\theta (x - \gamma)) + 36 \boldsymbol{\sigma}(\theta (x - \gamma))^2 - 24 \boldsymbol{\sigma}(\theta (x - \gamma))^3].
\end{align*}
\end{definition}
For proof, see McKenna \cite{sigmoid-diff}. We also define the following hypothesis,

\begin{hypothesis} \label{initial-growth-speed}
Initial observed growth speed is approximately equal to the maximum growth speed observed at the latter stage of growth.
\end{hypothesis}

From our data, we observe initial high growth speed, then a decline (due to the shock of being seeded), then a recovery and a final decline (due to limited space). Although, the initial growth speed and the maximum growth speed at latter stage of growth may not be equal (this can be observed by examining sub-figures \ref{growth-speed0} and  \ref{growth-speed1}), in an attempt to reduce the number of independent variables, we assume growth speed at the seeding time and maximum growth speed are approximately equal. Also, if the maximum growth speed at the latter growth phase is optimal, then it is unlikely that the growth speed at the seeding time can exceed the optimal growth speed significantly. Thus, justifying hypothesis \ref{initial-growth-speed}.\\

Above observations lead us to conclude that the cell growth in a confined space can be modelled by combining both the sigmoid function and the exponential of the sigmoid function. Thus, population  growth model and population growth speed model can be respectively expressed as follows,
\begin{align}
\log(y_n(x)) & = \log\left(y_n(0)\right)  + \beta_n \left[\boldsymbol{\sigma}(\theta_n (x - \gamma_n)) + \epsilon_n\log (\boldsymbol{\sigma}(c x))\right] ~\text{and}\label{growth-model}\\
\frac{d\log(y_n(x))}{dx} & = \omega_n \left[\frac{4}{\theta_n}\boldsymbol{\sigma}^{(1)}(\theta_n (x - \gamma_n)) + 2\delta_n (1 - \boldsymbol{\sigma}(c x))  \right], \label{growth-speed-model}
\end{align}
where $y_n(0)$ is the number of cells initially seeded, $\beta_n$ is a capacity, $\omega_n$ is a rate, $\gamma_n$ is a time and $\theta_n$, $\epsilon_n$, $\delta_n$ and $c$ are constants that are yet to be determined.\\

Given that there is only one cell seeded in the well (i.e. $y_n(0) = 1$, as we dealing with monoclones), equation (\ref{growth-model}) implies that $\epsilon_n = \frac{\boldsymbol{\sigma}(-\theta_n \gamma_n)} {\log(2) }$. Now, comparing equations (\ref{growth-model}) and (\ref{growth-speed-model}) imply that $ \theta_n = 4 \frac{\omega_n}{\beta_n} $ and $ c = \rho_n \theta_n$, where  $ \rho_n = \frac{\delta_n}{2\epsilon_n}$. Furthermore, hypothesis \ref{initial-growth-speed} and equation (\ref{growth-speed-model}) imply that $(2 \boldsymbol{\sigma}(\rho_n \theta_n \gamma_n )  - 1 ) \delta_n \approx  1 - 4\boldsymbol{\sigma}(-\theta_n \gamma_n)(1-\boldsymbol{\sigma}(-\theta_n \gamma_n)) $. Finally, observing that $ \boldsymbol{\sigma}(\rho_n \theta_n \gamma_n ) \approx 1$, we may assume that $\delta = 1 - 4\boldsymbol{\sigma}(-\theta_n \gamma_n)(1 -\boldsymbol{\sigma}(-\theta_n \gamma_n)) $. Collecting everything results in the following theorem,

\begin{Theorem}[Population Growth Model]\label{full-cell-growth-model}
Assuming hypothesis \ref{initial-growth-speed}, cell growth in a confined space can be modelled by the following equation,
\begin{align}
\log(y_n(x)) = \beta_n [\alpha_n + \boldsymbol{\sigma}(\theta_n (x - \gamma_n)) + \epsilon_n\log (\boldsymbol{\sigma}(\rho_n \theta_n x))], \label{cell-growth-model}
\end{align}
where $y_n(x)$ is the cell count, $x$ is the time, $\boldsymbol{\sigma}(\cdot)$ is the sigmoid function (definition \ref{sigmoid-diff}),
\begin{align*}
\theta_n & = 4 \frac{\omega_n}{\beta_n} ,\\
\rho_n &= \frac{\delta_n}{2\epsilon_n} ,
\end{align*}
$\beta_n$ is the population growth capacity, $\omega_n$ is population growth rate, $\gamma_n$ is the population growth time, 
\begin{align*}
\epsilon_n = \frac{\boldsymbol{\sigma}(-\theta_n \gamma_n)} {\log(2) }
\end{align*}
is the incipient population growth capacity coefficient,
\begin{align*}
\delta_n =  1 - 4\boldsymbol{\sigma}(-\theta_n \gamma_n)(1- \boldsymbol{\sigma}(-\theta_n \gamma_n)) 
\end{align*}
is the incipient population growth rate coefficient, 
\begin{align*}
\alpha_n  =\mathbb{E}\left[\frac{1}{ \beta_n} \log(y_n(x)) - \boldsymbol{\sigma}(\theta_n (x - \gamma_n)) - \epsilon_n\log (\boldsymbol{\sigma}(\rho_n \theta_n x))\mid (\beta_n, \gamma_n, \omega_n) \right]
\end{align*}
is the minimum population coefficient and $\mathbb{E}(\cdot)$ is the expectation operator, and where $\beta_n$, $\omega_n$ and $\gamma_n$ are the only independent variables of the model, and $x$ is the only independent variable of the dataset.
\end{Theorem}

Interpretation of theorem \ref{full-cell-growth-model} is as follows. $\beta_n$ defines the capacity  of the population to colonise the given environment, where a higher capacity implies a higher likelihood of colonisation as a measure of the log cell count. $\omega_n$ defines population's maximum growth rate, where a higher rate implies a higher rate of cell multiplication. $\gamma_n$ defines the time for the population to reach its maximum growth rate, where a lower time results in a shorter time to reach its maximum growth rate. $\epsilon_n$ defines the shock of being seeded, where a small coefficient results in a very resilient cell when seeded. $\delta_n$ defines the rate of recovery after been seeded, where a higher coefficient results in a faster recovery after being seeded. $\alpha_n$ defines the likelihood of multiples sells seeded. As we are modelling the growth of monoclones, $\alpha_n$ should be zero. Thus, any non-zero value implies a likelihood of multiple cells been seeded, and thus, giving us a good measure of how trustworthy the monoclone sample is. Definition of $\alpha_n$ states that  $\alpha_n$ is the expected value of the difference between observed values and modelled values, given that the independent variables are already found. This definition may seems rather abstract; however, we explain this in more detail during our numerical modelling section (see section \ref{numerical-models}). Also, why we normalise $\alpha_n$ to be independent of the dimensions of other variables is explained in section \ref{sec-titer}, when discussing practical applications of the variables in predicting titer measurements and viability.\\

Although, we derived equation (\ref{cell-growth-model}) based on the growth of CHO cells in circular wells (i.e. 96-well microplates), we predict that the model will hold true for other similar  cell types (e.g. HEK, HeLa, Jurkat, etc.) and other similar well types (e.g. 348-well microplates with squircle wells). Furthermore, taking the first and second order derivative of the equation (\ref{cell-growth-model}), we arrive at the following corollary,

\begin{corollary}\label{full-cell-growth-speed-model}
Theorem \ref{full-cell-growth-model} implies that the population growth speed and population growth acceleration can respectively expressed as follows,
\begin{align}
\frac{d\log(y_n(x))}{dx}  &= \omega_n \left[ \frac{4}{\theta_n} \boldsymbol{\sigma}^{(1)}(\theta_n (x - \gamma_n))+ 2\delta_n (1 - \boldsymbol{\sigma}(\rho_n \theta_n x))  \right]~\text{and} \label{cell-growth-speed} \\
\frac{d^2\log(y_n(x))}{dx^2} &= \omega_n \theta_n\left[ \frac{4}{(\theta_n)^2} \boldsymbol{\sigma}^{(2)}(\theta_n (x - \gamma_n)) - 2\frac{\delta_n}{\theta_n} \boldsymbol{\sigma}^{(1)}(\rho_n \theta_n x)  \right],\label{cell-growth-acceleration}
\end{align}
where $\boldsymbol{\sigma}^{(1)}(\cdot)$ and  $\boldsymbol{\sigma}^{(2)}(\cdot)$ are the first and second order derivatives of the sigmoid function (definition \ref{sigmoid-diff}). Equation (\ref{cell-growth-model}) implies that the cell culture's generation can be expressed as follows,
\begin{align*}
n(x)  = \frac{1}{\log(2)} \log\left(\frac{y_n(x)}{y_n(0)}\right),
\end{align*}
and the maximum number of cells that originate from a monoclone that a confined space can theoretically support can be calculated as follows,
\begin{align*}
\lim_{x\to\infty}\frac{y_n(x)}{y_n(0)}= \exp(\beta_n).
\end{align*}
Equation (\ref{cell-growth-speed}) implies that the doubling time of the population can be expressed as follows,
\begin{align*}
T_{D}(x)  = \frac{\log(2)}{\frac{d\log(y_n(x))}{dx}},
\end{align*}
and the doubling rate of the population can expressed as follows,
\begin{align*}
R_{D}(x)  = \frac{1}{\log(2)} \frac{d\log(y_n(x))}{dx}.
\end{align*}
Equation (\ref{cell-growth-acceleration}) implies that maximum population growth acceleration observed at $ t_{n0} = \gamma_n +  \frac{\beta_n}{4\omega_n}\boldsymbol{\sigma}^{-1}(z_0)$, where $z_0 = \frac{1}{2} \left( 1 - \frac{\sqrt{3}}{3}\right)$ and $\boldsymbol{\sigma}^{-1}(z) = \log\left( \frac{z}{1-z} \right)$ is the inverse of the sigmoid function (i.e. the logit function). Also, the point that $\frac{d\log(y_n(x))}{dx}$ attains a minimum in the interval  $0$ and $ t_{n0}$ can be interpreted as the recovery time, i.e. time for cells to recover after the shock of from being seeded. Furthermore, $t_{n1} =  \gamma_n +  \frac{\beta_n}{4\omega_n}\boldsymbol{\sigma}^{-1}(z_1) $ is the latter maximum deceleration time, where $z_1 = \frac{1}{2} \left( 1 + \frac{\sqrt{3}}{3}\right) $.
\end{corollary}

\subsection{Numerical Modelling \label{numerical-models}}

Given a dataset $(x, \log(y_n))$, the equation (\ref{cell-growth-model}) is difficult to numerically model on its own to get consistent results. In an attempt to get consistent results for the entire dataset, i.e. unique and finite $(\beta_n, \gamma_n, \omega_n, \alpha_n)$-set, we precent the following algorithm.\\

\emph{Step 1:} Normalise the dataset. First normalise the time data points as $x_{\text{norm}} =\frac{x}{ x_{\text{max}} }$, where $x_{\text{max}} =\max(x)$ and $\max(\cdot)$ is the maximum element operator. Then using \emph{NumPy} \emph{polyfit} function with \emph{deg}$=1$ \cite{numpy.polyfit}, fit a line of best fit to the log cell count data points as $\log(y_n) = a + b x_{\text{norm}}$. Now, using the coefficients $a$ and $b$, normalise the cell count data points as $\log(y_n)_{\text{norm}} =\frac{\log(y_n) - \log(y_n)_{\text{min}}}{ \log(y_n)_{\text{max}} }$, where $\log(y_n)_{\text{min}} = a - \text{RMSE}$, $\log(y_n)_{\text{max}} = b + 2 \text{RMSE}$ and $\text{RMSE}$ is the root mean square error of the linear regression.\\

\emph{Step 2:} Find lower and upper bounds for the parameters. Using  \emph{SciPy} \emph{curve\_fit} function with \emph{maxfev}$=10,000$ \cite{scipy.curvefit}, fit the normalised data to the following equation,
\begin{align*}
\log(y_n)_{\text{norm}} = b_n \boldsymbol{\sigma}(4 d_n (x_{\text{norm}} - c_n)),
\end{align*}
where $b_n$, $c_n$ and $d_n$ are bounded below by $0$ and above by $2$.\\

\emph{Step 3:} Find normalised population growth parameters. Using the bounds $0<\beta_n^0 <b_n$, $c_n< \gamma_n^0 < 2$ and $0 < \omega^0_n < b_n d_n$, and using \emph{SciPy} \emph{curve\_fit} function with \emph{maxfev}$=10,000$  \cite{scipy.curvefit}, fit the data to the following equation
\begin{align*}
\log(y_n)_{\text{norm}} = \beta^0_n \big[\boldsymbol{\sigma}\big(\theta^0_n(x_{\text{norm}} - \gamma^0_n)\big) + \epsilon^0_n\log \big(\boldsymbol{\sigma}\big(\rho^0_n \theta^0_n x_{\text{norm}}\big)\big)\big],
\end{align*}
where $\epsilon^0_n$, $\theta^0_n$ and $\rho^0_n$  dependent variables of $\beta_n^0$, $ \gamma_n^0$ and $\omega^0_n$, which are defined in theorem \ref{full-cell-growth-model}. With $\beta_n^0$, $ \gamma_n^0$ and $\omega^0_n$, find $\alpha_n^0$  as follows,
\begin{align*}
\alpha^0_n  = ~& \mathbb{E}\left[ \log(y_n)_{\text{norm}} -  \beta_n^0[\boldsymbol{\sigma}\big(\theta_n^0 \big(x_{\text{norm}} - \gamma_n^0\big)\big) + \epsilon_n^0\log \big(\boldsymbol{\sigma}\big(\rho_n^0 \theta_n^0 x_{\text{norm}}\big)\big)\big]\right] \\
& + \frac{\log(y_n)_{\text{min}}}{\log(y_n)_{\text{max}}}.
\end{align*}
Note that how we had to find $\beta_n^0$, $ \gamma_n^0$ and $\omega^0_n$ first, in order to calculate $\alpha^0_n$; thus, justifying the definition of $\alpha_n$.\\

\emph{Step 4:} Unnormalise the population growth parameters as $\alpha_n  = \frac{1}{\beta_n^0}\alpha^0_n$, $\beta_n  = \log(y_n)_{\text{max}} \beta_n^0 $,  $\gamma_n  = x_{\text{max}} \gamma_n^0 $ and  $\omega_n  = \frac{\log(y_n)_{\text{max}} }{x_{\text{max}}} \omega_n^0 $.\\

As what we describe is rather abstract, following is a working \emph{python} code so the reader may replicate our algorithm.

\begin{python}
# Sigmoid function
def sigmoid_fn(x):
  sigmoid_x = 1.0 / (1.0 + np.exp(-x))
  return sigmoid_x
\end{python}

\begin{python}
# For finding boundary values
def growth_simple_fit_fn(x, b, c, d):
  log_y = b * sigmoid_fn(4 * d * (x - c))
  return log_y
\end{python}

\begin{python}
# For finding finding growth variables
def growth_fit_fn(x, beta, gamma, omega):

  theta = 4 * omega / beta
  sigma = sigmoid_fn(-theta * gamma)

  e = beta * sigma / np.log(2)
  f = 0.5 * np.log(2) * theta * ((1/sigma) - 4  + 4 * sigma)

  u = theta * (x - gamma)
  v = f * x

  log_y = beta * sigmoid_fn(u) + e * np.log(sigmoid_fn(v))
  return log_y
\end{python}

\begin{python}
# Fit the cell-data to the growth_fit_fn
processed_data = []

max_days = 0
max_itter = 10000

x_epsilon = 1e-7
x_upper_bound = 2

len_dataset = len(cell_data)
for datapoint in range(len_dataset):

  try:
    cell_data_point = un_processed_data [ datapoint ]

    cell_name = cell_data_point[0]
    x_data = np.array(cell_data_point[1]) # time axis
    y_data = np.array(cell_data_point[2]) # cell population

    y_data = np.where(y_data > 0, y_data, 1.0)
    y_data = np.log(y_data_raw) # log cell population

    ## normalise dataset
    x_max = np.max(x_data)
    x_data = x_data / x_max

    y_stats = np.polyfit(x_data, y_data, 1, full = True)
    y_a = y_stats[0][1]
    y_b = y_stats[0][0]
    y_rmse = np.sqrt(y_stats[1][0] / len(y_data))

    y_min = y_a - y_rmse
    y_max = y_b + 2 * y_rmse

    y_data = y_data - y_min
    y_data = y_data / y_max

    ## find bounds for the variables
    coefficients_initial = curve_fit(growth_simple_fit_fn, x_data, y_data, bounds = (x_epsilon, x_upper_bound - x_epsilon), maxfev = max_itter)[0]

    x_beta = coefficients_initial[0]
    x_gamma = coefficients_initial[1]
    x_omega = coefficients_initial[0] * coefficients_initial[2]

    lower_bound = [0, x_gamma, 0]
    upper_bound = [x_beta, x_upper_bound, x_omega]

    ## find growth variables
    coefficients_fit = curve_fit(growth_fit_fn, x_data, y_data, bounds = (lower_bound, upper_bound), maxfev = max_itter)[0]

    alpha_beta = np.mean(y_data - growth_fn(x_data, 0, coefficients_fit[0], coefficients_fit[1], coefficients_fit[2])) + (y_min / y_max)

    ## un-normalise variables
    alpha = alpha_beta / coefficients_fit[0]
    beta = y_max * coefficients_fit[0]
    gamma = x_max * coefficients_fit[1]
    omega = y_max * coefficients_fit[2] / x_max

    coefficients = [alpha, beta, gamma, omega]

    ## append data
    processed_data.append([cell_name, coefficients])

  except:
    print("Unable to process data")
\end{python}
where \emph{un\_processed\_data} is a collection of [\emph{cell name, time array, cell count array}]-points and \emph{processed\_data} is a collection of [\emph{cell name}, $[\alpha_n, \beta_n, \gamma_n, \omega_n]$]-points. For a sample dataset, along with a working algorithm, including some interesting plots, please see the link in the footnote\footnote{Population growth model \emph{Colab} notebook with a sample dataset, a working algorithm and plots: 
\url{https://drive.google.com/file/d/1yDLHhs3p2oQ1Hd2--W6luJxUc0wfV6y7/view?usp=drive_link}} for a \emph{Colab} notebook.

\subsection{Experimental Results} \label{n-plots}

In this section, we fit the our population growth model  (theorem \ref{full-cell-growth-model}) to a dataset of observations of 166 CHO monoclones in 96-well microplates over a 15 day period, i.e. the set \emph{CHO2023}, where the data is provided by Wheeler Bio, Inc. \cite{wheelerbio} and the data is processed with DeepInsight\textsuperscript{\textregistered} cell analysis software \cite{tekor}.\\

\begin{figure}[!h]
     \centering
     \begin{subfigure}[b]{0.3\linewidth}
         \raggedleft
         \includegraphics[height=\textwidth]{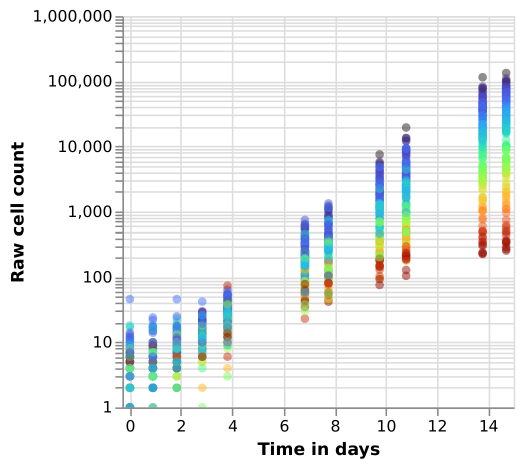}
         \caption{Raw cell count.}
         \label{cell-count0x}
     \end{subfigure}
     \hfill
     \begin{subfigure}[b]{0.3\textwidth}
         \centering
         \includegraphics[height=\textwidth]{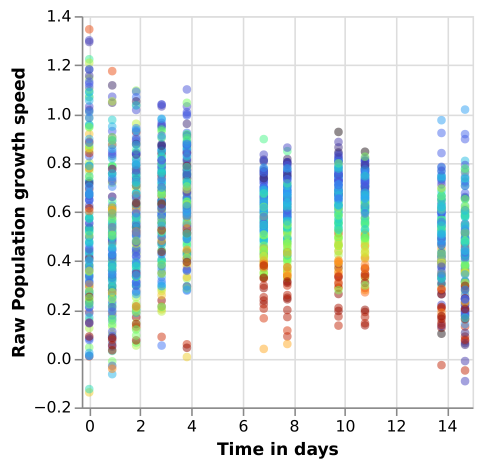}
         \caption{Raw speed.}
         \label{growth-speed0x}
     \end{subfigure}
     \hfill
     \begin{subfigure}[b]{0.3\textwidth}
         \centering
         \includegraphics[height=\textwidth]{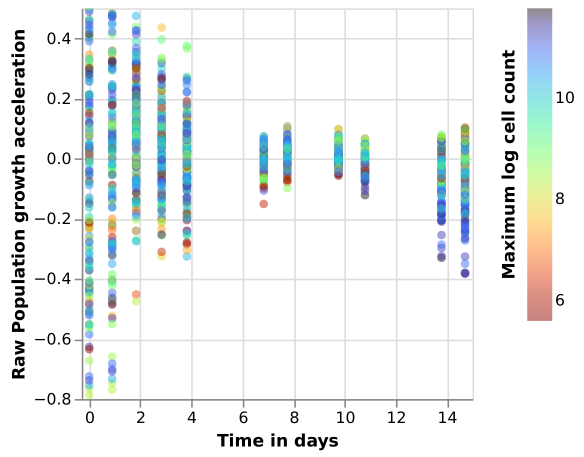}
         \caption{Raw acceleration.}
         \label{growth-acceleration0x}
     \end{subfigure}
	\hfill
     \begin{subfigure}[b]{0.09\textwidth}
     \end{subfigure}
     \hfill
     \begin{subfigure}[b]{0.3\textwidth}
         \centering
         \includegraphics[height=\textwidth]{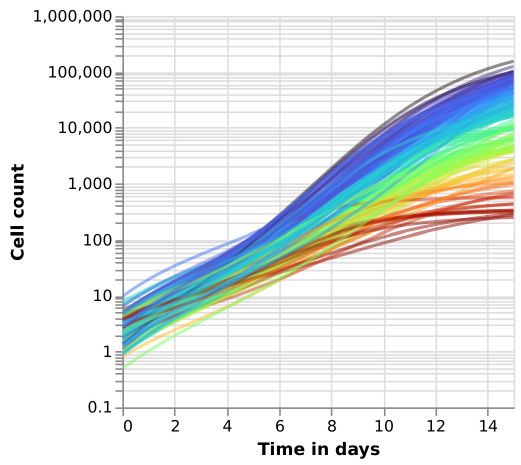}
         \caption{Cell count.}
         \label{modelled-cell-count1x}
     \end{subfigure}
     \hfill
     \begin{subfigure}[b]{0.3\textwidth}
         \centering
         \includegraphics[height=\textwidth]{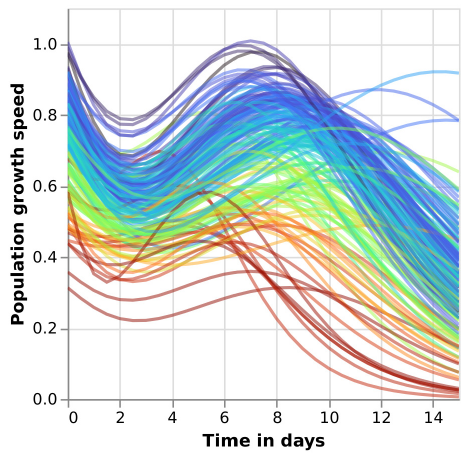}
         \caption{Speed.}
         \label{modelled-growth-speed1x}
     \end{subfigure}
     \hfill
     \begin{subfigure}[b]{0.3\textwidth}
         \centering
         \includegraphics[height=\textwidth]{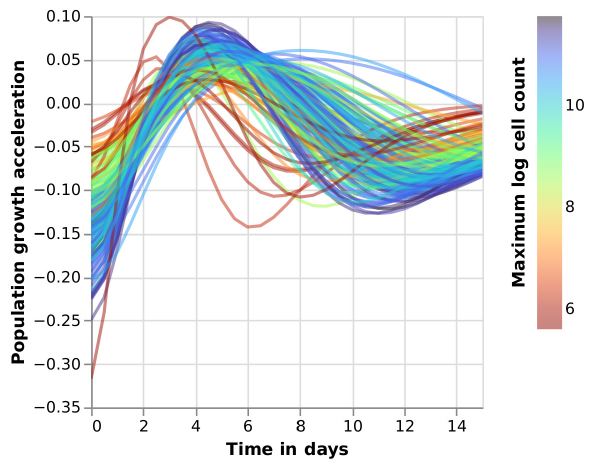}
         \caption{Acceleration.}
         \label{modelled-growth-acceleration1x}
     \end{subfigure}
	\hfill
     \begin{subfigure}[b]{0.09\textwidth}
     \end{subfigure}
     \hfill
        \caption{Population growth patterns (cell count): raw vs modelled.}
        \label{experiment-x}
\end{figure}

Figure \ref{experiment-x} shows cell count, population growth speed and population growth acceleration, along with what is implied by our model, where by \emph{raw} means no model is assumed. As the reader can see that our population growth equation (equation (\ref{cell-growth-model})) is far more accurate at capturing the complex behaviour of the growth speed and the growth acceleration than standard sigmoid function (equation (\ref{sigmoid})). Sub-figures \ref{growth-speed0x} and \ref{modelled-growth-speed1x} show that the initial growth speed rapidly declines and then recovers back to its optimal population growth rate (i.e. $\omega_n$) at the population growth time (i.e. $\gamma_n$), and finally declining as the well approaches full confluency. Sub-figures \ref{growth-acceleration0x} and \ref{modelled-growth-acceleration1x} imply that the rapid decline in the growth rate after being seeded (i.e. shock of being seeded) is captured by $\epsilon_n$ (i.e. incipient population growth capacity coefficient) and its recovery $\delta_n$ (i.e. incipient population growth rate coefficient). In other words, a small $\epsilon_n$ and a large $\delta_n$  result in a better incipient population growth.\\

\begin{figure}[h!]
     \centering
     \begin{subfigure}[b]{0.3\textwidth}
         \centering
         \includegraphics[height=\textwidth]{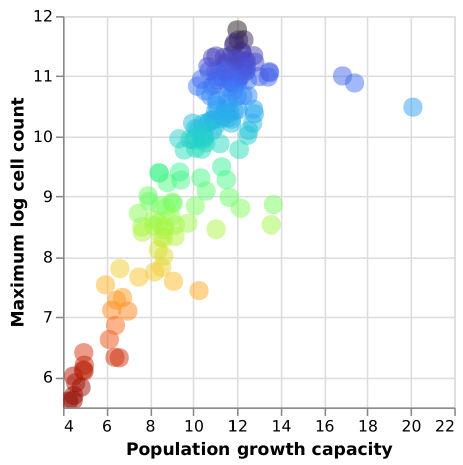}
         \caption{$\beta_n$}
         \label{beta-n}
     \end{subfigure}
     \hfill
     \begin{subfigure}[b]{0.3\textwidth}
         \centering
         \includegraphics[height=\textwidth]{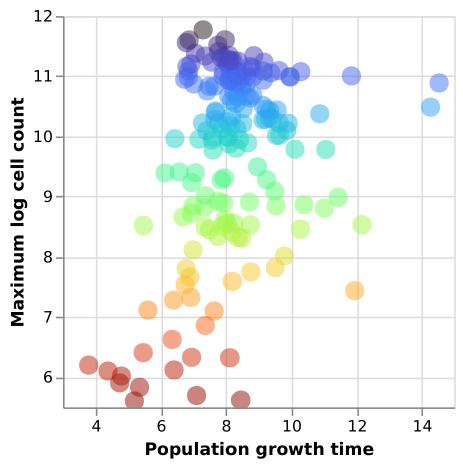}
         \caption{$\gamma_n$}
         \label{gamma-n}
     \end{subfigure}
     \hfill
     \begin{subfigure}[b]{0.3\textwidth}
         \centering
         \includegraphics[height=\textwidth]{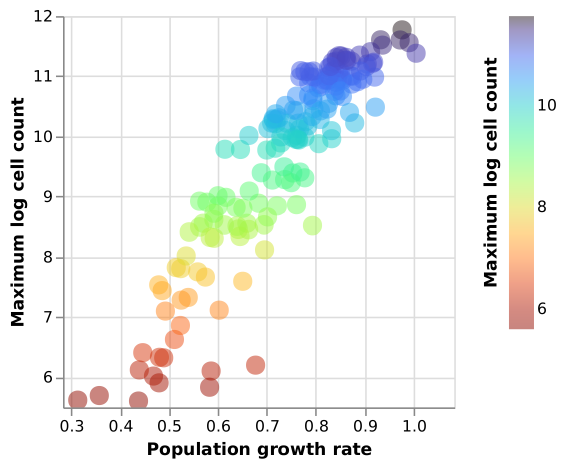}
         \caption{$\omega_n$}
         \label{omega-n}
     \end{subfigure}
		\hfill
     \begin{subfigure}[b]{0.09\textwidth}
     \end{subfigure}
        \caption{Population growth variables.}
        \label{n-var}
\end{figure}

Figure \ref{n-var} shows the population growth variables with respect to maximum of the log cell count, were the blue-shifted dots represent monoclone samples with high final cell counts (upper limit is $129,508$ cells) and the red-shifted dots represent monoclone samples with low final cell counts (lower limit is $272$ cells). Figure \ref{beta-n} implies that the population growth capacity is positively correlated with the final  log cell count ($\rho = 0.838$, where $\rho$ is the correlation coefficient and all correlation coefficients are rounded to 3 decimal points) and  figure \ref{omega-n} implies that the population growth rate is positively correlated with the final log cell count ($\rho = 0.909$). Figure \ref{gamma-n}, implies a positive correlation between the population growth time and the final log cell count ($\rho = 0.347$); however, this correlation becomes negative for monoclones with a final cell count above $\log(11)$, implying a lower population growth time results in a greater final cell count for blue-shifted monoclones, justifying our intuitive understanding of $\gamma_n$.\\

\begin{figure}[h!]
     \centering
     \begin{subfigure}[b]{0.30\textwidth}
         \centering
         \includegraphics[height=\textwidth]{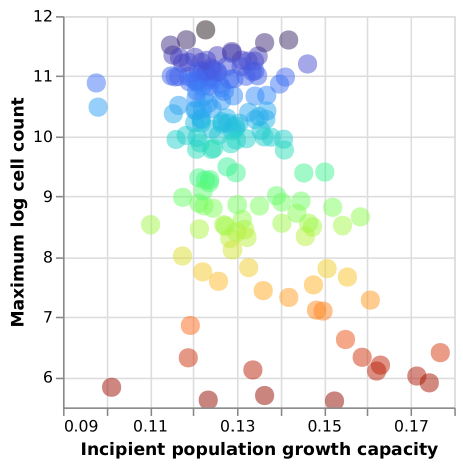}
         \caption{$\epsilon_n$}
         \label{epsilon-n}
     \end{subfigure}
     \hfill
     \begin{subfigure}[b]{0.30\textwidth}
         \centering
         \includegraphics[height=\textwidth]{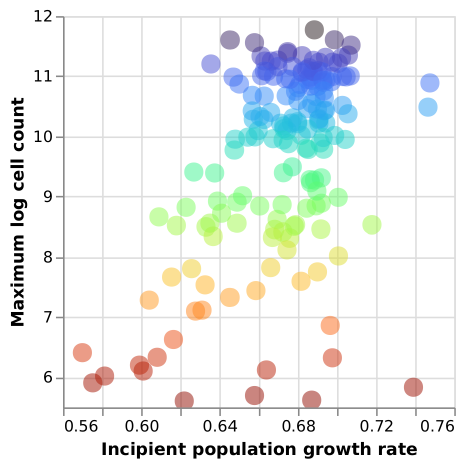}
         \caption{$\delta_n$}
         \label{delta-n}
     \end{subfigure}
     \hfill
     \begin{subfigure}[b]{0.30\textwidth}
         \centering
         \includegraphics[height=\textwidth]{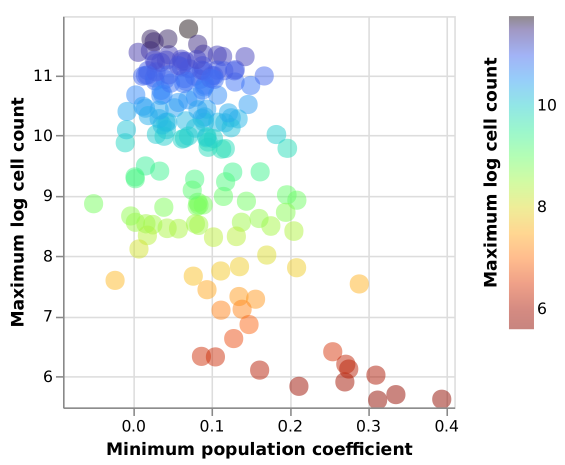}
         \caption{$\alpha_n$}
         \label{alpha-n}
     \end{subfigure}
	\hfill
     \begin{subfigure}[b]{0.09\textwidth}
     \end{subfigure}
        \caption{Incipient population growth coefficients.}
        \label{n-var-2}
\end{figure}

Figure \ref{n-var-2} shows the incipient growth coefficients with respect to maximum of the log cell count. Figure \ref{epsilon-n} implies that incipient growth capacity coefficient is negatively correlated with the final log cell count ($\rho = -0.569$), figure \ref{delta-n} implies that incipient growth rate coefficient is positively correlated with the final log cell count ($\rho = 0.497$) and  figure \ref{alpha-n} implies that minimum population coefficient is negatively correlated with the final log cell count ($\rho = -0.502$) . \\

\begin{figure}[h!]
     \centering
     \begin{subfigure}[b]{0.49\textwidth}
         \centering
         \includegraphics[height=\textwidth]{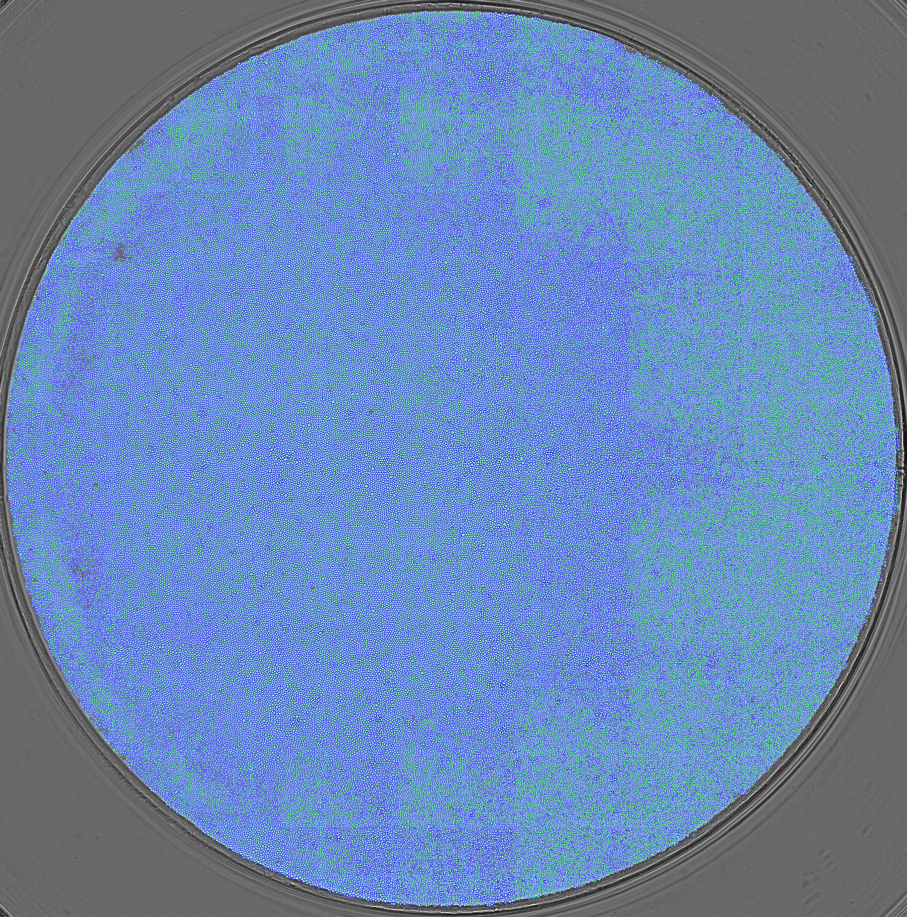}
         \caption{Day 14.}
         \label{day14}
     \end{subfigure}
     \hfill
     \begin{subfigure}[b]{0.49\textwidth}
         \centering
         \includegraphics[height=\textwidth]{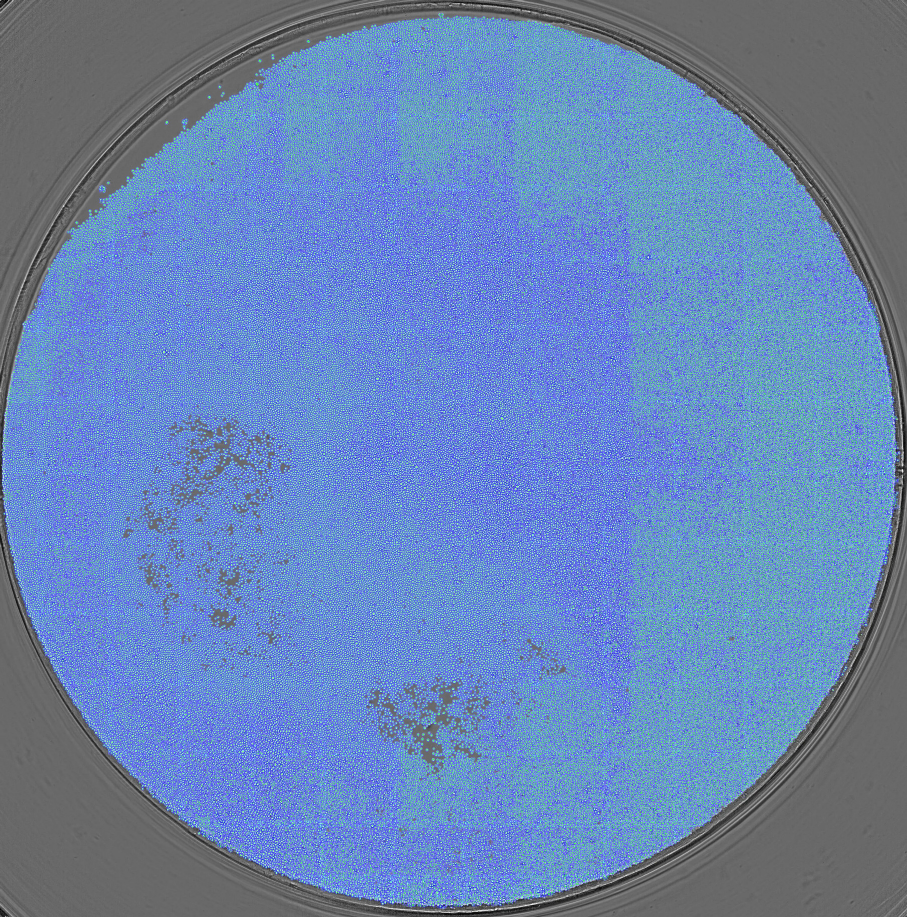}
         \caption{Day 15.}
         \label{day15}
     \end{subfigure}
	  \caption{CHO monoclone $cho\#11771$: full well.}
        \label{full}
\end{figure}

\begin{figure}[!h]
     \centering
     \begin{subfigure}[b]{0.49\textwidth}
         \centering
         \includegraphics[height=\textwidth]{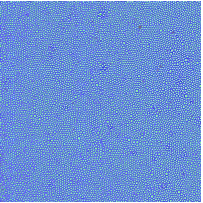}
         \caption{Day 14.}
         \label{day14-zoom}
     \end{subfigure}
     \hfill
     \begin{subfigure}[b]{0.49\textwidth}
         \centering
         \includegraphics[height=\textwidth]{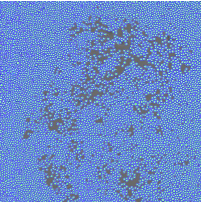}
         \caption{Day 15.}
         \label{day15-zoom}
     \end{subfigure}
        \caption{CHO monoclone $cho\#11771$: zoomed in to the anomaly.}
        \label{zoom}
\end{figure}

From our dataset, we find the monoclone $cho\#11771$ has the greatest maximum cell count, where this maximum cell count is observed on day 14, with near $100\%$ confluence and a cell count of $129,508$. However, on day 15,  we observe the cell count declined to $122,293$, and also a decline in confluency. Should one examine the original images of the well in more depth (see figure \ref{full}), one can see that the after reaching $100\%$ confluency in day 14, the cells starts to form vertical colonies. In fact, should one examine sub-figures \ref{day14-zoom} and \ref{day15-zoom} (where the cropping window is $(0.11,0.33) \times (0.34 , 0.56)$, measured from the bottom left-hand corner of the original image), one can see that reduction in confluence from day 14 to day 15, when cells coalesced to form vertical colonies. TEK Optima Research Ltd. DeepInsight\textsuperscript{\textregistered} cell analysis software \cite{tekor} can only locate visible cells, and thus, cells that are buried under vertical colonies missed in the final count, resulting in the appearance of the decline in overall cell count. In such situations, one can use the population growth equation to estimate the number of cells that the well can theoretically support, i.e. theorem \ref{full-cell-growth-model} implies that for monoclone $cho\#11771$ the cell count on day 15 is $143,622$. Should one recall corollary \ref{full-cell-growth-speed-model}, one finds a theoretical maximum cell count of $164,872$ (i.e. $\exp(\beta_n)$ where $\beta_n = 12.012927$) for monoclone $cho\#11771$. What this shows is that our model captured growth patterns to a degree of accuracy where it can predict the future growth behaviour of the cells in a well.

\section{Confluence}

Another important measure in cell line development is the measure of confluence, i.e. the measure of the surface area covered by cells. Following similar reasonings as in section (\ref{cell-count}) we arrive at the following theorem,

\begin{Theorem}[Confluence Growth Model]\label{cell-confluence}
Confluence growth in a confined space can be modelled by the following equation,
\begin{align}
\log(y_c(x)) =  \beta_c  [\alpha_c + \boldsymbol{\sigma}(\theta_c (x - \gamma_c))] , \label{cell-confluence-eqn}
\end{align}
where $y_c(x)$ is the confluence in pixels, $x$ is the time, $\boldsymbol{\sigma}(\cdot)$ is the sigmoid function (definition \ref{sigmoid-diff}),
\begin{align*}
\theta_c & = 4 \frac{\omega_c}{\beta_c} ,
\end{align*}
$\beta_c$ is the confluence growth capacity, $\omega_c$ is confluence growth rate, $\gamma_c$ is the confluence growth time, 
\begin{align*}
\alpha_c  =\mathbb{E}\left[ \frac{1}{\beta_c} \log(y_c(x)) -  \boldsymbol{\sigma}(\theta_c (x - \gamma_c)) \mid (\beta_c, \gamma_c, \omega_c) \right] 
\end{align*}
is the minimum confluence coefficient, $\mathbb{E}(\cdot)$ is the expectation operator, and where $\beta_c$, $\omega_c$ and $\gamma_c$ are the only independent variables of the model, and $x$ is the only independent variable of the dataset.
\end{Theorem}

Interpretation of theorem \ref{cell-confluence} is as follows. $\beta_c$ defines the capacity of the population to colonise the given environment, where a higher capacity implies a higher likelihood of colonisation as a measure of the log area. $\omega_c$ defines maximum confluence growth rate, where a higher rate implies a higher rate of cell multiplication and/or larger cell area. $\gamma_c$ defines the time for the population to reach $50\%$ confluence, where a lower time results in a shorter time to reach $50\%$ confluence. $\alpha_c$ defines the $\log$ of the initial confluence normalised by $\beta_c$. A large value implies a likelihood of multiple cells been seeded, and thus, giving us another good measure of how trustworthy the monoclone sample is.\\

Now, taking the first and second order derivative of the equation (\ref{cell-confluence-eqn}), we arrive at the following corollary, 

\begin{corollary}\label{c-c}
Theorem \ref{cell-confluence} implies that the confluence growth speed and confluence growth acceleration can respectively expressed as follows,
\begin{align}
\frac{d\log(y_c(x))}{dx}  &= \omega_c \left[ \frac{4}{\theta_c} \boldsymbol{\sigma}^{(1)}(\theta_c (x - \gamma_c)) \right]~\text{and}\nonumber  \\
\frac{d^2\log(y_c(x))}{dx^2} &= \omega_n \theta_c\left[ \frac{4}{(\theta_c)^2} \boldsymbol{\sigma}^{(2)}(\theta_c (x - \gamma_c))  \right],\label{con-acceleration}
\end{align}
where $\boldsymbol{\sigma}^{(1)}(\cdot)$ and  $\boldsymbol{\sigma}^{(2)}(\cdot)$ are the first and second order derivatives of the sigmoid function (definition \ref{sigmoid-diff}). Equation (\ref{con-acceleration}) implies that maximum confluence growth acceleration observed at $ t_{c0} = \gamma_c +  \frac{\beta_c}{4\omega_c}\boldsymbol{\sigma}^{-1}(z_0)$, where $z_0 = \frac{1}{2} \left( 1 - \frac{\sqrt{3}}{3}\right)$ and $\boldsymbol{\sigma}^{-1}(z) = \log\left( \frac{z}{1-z} \right)$ is the inverse of the sigmoid function (i.e. the logit function). Also, $t_{c1} =  \gamma_c +  \frac{\beta_c}{4\omega_c}\boldsymbol{\sigma}^{-1}(z_1) $ is the latter maximum deceleration time, where $z_1 = \frac{1}{2} \left( 1 + \frac{\sqrt{3}}{3}\right) $ is the passage confluence point  ($\approx 78.9\%$).
\end{corollary}

\subsection{Numerical Modelling}

As in section \ref{numerical-models}, given a dataset $(x, \log(y_c))$, to find a $(\beta_c, \gamma_c, \omega_c, \alpha_c)$-set, we precent the following algorithm.\\

\emph{Step 1:} Normalise the dataset. First normalise the time data points as $x_{\text{norm}} =\frac{x}{ x_{\text{max}} }$, where $x_{\text{max}} =\max(x)$. Then using \emph{NumPy} \emph{polyfit} function with \emph{deg}$=1$ \cite{numpy.polyfit}, fit a line of best fit to the log confluence data points as $\log(y_c) = a + b x_{\text{norm}}$. Now, using the coefficients $a$ and $b$, normalise the confluence data points as $\log(y_c)_{\text{norm}} =\frac{\log(y_c) - \log(y_c)_{\text{min}}}{ \log(y_c)_{\text{max}} }$, where $\log(y_c)_{\text{min}} = a - \text{RMSE}$, $\log(y_c)_{\text{max}} = b + 2 \text{RMSE}$.\\

\emph{Step 2:} Find lower and upper bounds for the parameters. Using  \emph{SciPy} \emph{curve\_fit} function with \emph{maxfev}$=10,000$ \cite{scipy.curvefit}, fit the normalised data to the following equation,
\begin{align*}
\log(y_c)_{\text{norm}} = b_c \boldsymbol{\sigma}(4 d_c (x_{\text{norm}} - c_c)) ,
\end{align*}
where $b_c$, $c_c$ and $d_c$ are bounded below by $0$ and above by $2$.\\

\emph{Step 3:} Find normalised confluence growth parameters. Using the bounds $0<\beta_c^0 <b_c$, $c_c< \gamma_c^0 < 2$ and $0 < \omega^0_c < b_c d_c$, and using \emph{SciPy} \emph{curve\_fit} with \emph{maxfev}$=10,000$  \cite{scipy.curvefit}, fit the data to the following equation,
\begin{align*}
\log(y_c)_{\text{norm}} = \beta^0_c \big[\boldsymbol{\sigma}\big(\theta^0_c(x_{\text{norm}} - \gamma^0_c)\big) - \boldsymbol{\sigma}\big(-\theta^0_c \gamma^0_c\big) \big],
\end{align*}
where $\theta^0_c = 4 \frac{\omega^0_c}{\beta^0_c}$. With $\beta_c^0$, $ \gamma_c^0$ and $\omega^0_c$, find $\alpha_c^0$  as follows,
\begin{align*}
\alpha^0_c  =\mathbb{E}\left[ \log(y_c)_{\text{norm}} -  \beta_c^0\boldsymbol{\sigma}\big(\theta_c^0 \big(x_{\text{norm}} - \gamma_c^0\big)\big)\right] + \frac{\log(y_c)_{\text{min}}}{\log(y_c)_{\text{max}}}.
\end{align*}

\emph{Step 4:} Unnormalise the confluence growth parameters as $\alpha_c  = \frac{1}{\beta_c^0}  \alpha^0_c$, $\beta_c  =\log(y_c)_{\text{max}} \beta_c^0 $,  $\gamma_c  = x_{\text{max}} \gamma_c^0 $ and  $\omega_c  = \frac{\log(y_c)_{\text{max}} }{x_{\text{max}}} \omega_c^0 $.\\

For a sample dataset, along with a working algorithm, please see the link in the footnote\footnote{Confluence growth model \emph{Colab} notebook with a sample dataset, a working algorithm and plots: 
\url{https://drive.google.com/file/d/1R9t1X2LYn71wkQU8oMyZhEImYD18-EgI/view?usp=drive_link}} for a \emph{Colab} notebook.

\subsection{Experimental Results}

In this section, we fit the our confluence growth model  (theorem \ref{full-cell-growth-model}) to \emph{CHO2023} dataset, where the confluence calculated as a number of pixels and the resolution is given to be $2\mu$m per pixel \cite{CellMetric0}.\\

\begin{figure}[h!]
     \centering
     \begin{subfigure}[b]{0.3\textwidth}
         \centering
         \includegraphics[height=\textwidth]{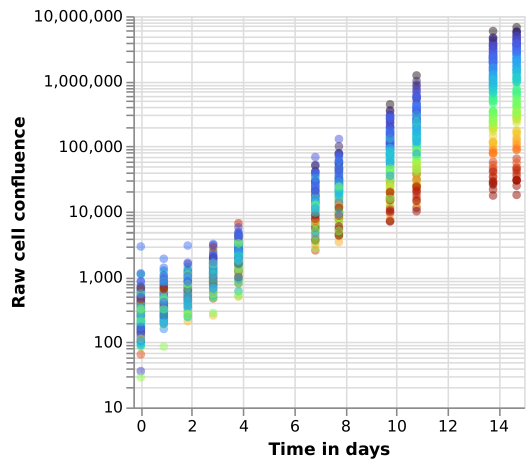}
         \caption{Raw confluence.}
     \end{subfigure}
     \hfill
     \begin{subfigure}[b]{0.3\textwidth}
         \centering
         \includegraphics[height=\textwidth]{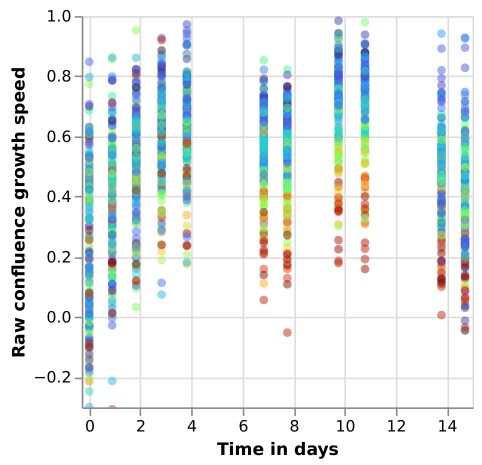}
         \caption{Raw speed.}
         \label{c-speed-raw}
     \end{subfigure}
     \hfill
     \begin{subfigure}[b]{0.3\textwidth}
         \centering
         \includegraphics[height=\textwidth]{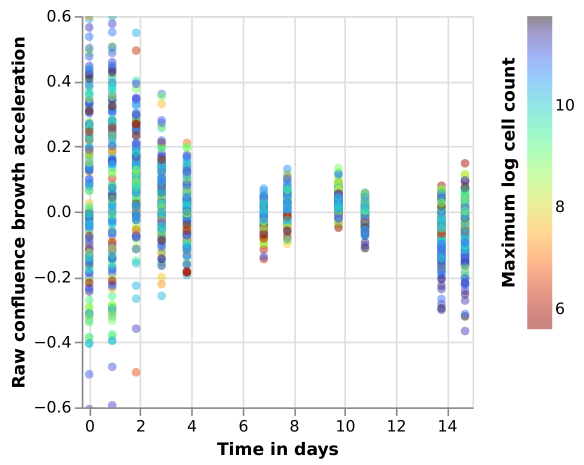}
         \caption{Raw acceleration.}
         \label{n-acc-raw}
     \end{subfigure}
     \hfill
     \begin{subfigure}[b]{0.09\textwidth}
     \end{subfigure}
     \hfill
     \begin{subfigure}[b]{0.3\textwidth}
         \centering
         \includegraphics[height=\textwidth]{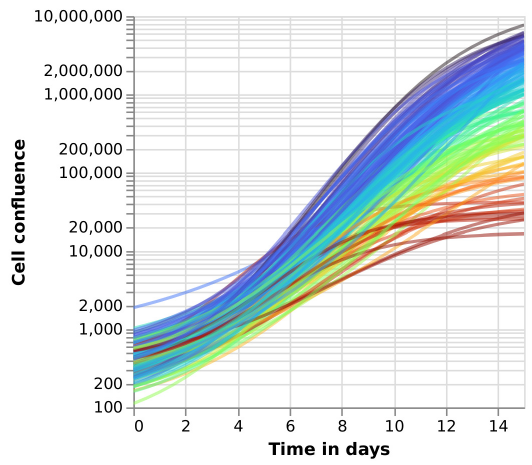}
         \caption{Confluence.}
     \end{subfigure}
     \hfill
     \begin{subfigure}[b]{0.3\textwidth}
         \centering
         \includegraphics[height=\textwidth]{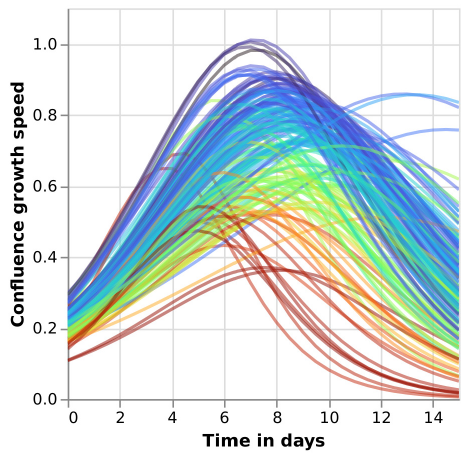}
         \caption{Speed.}
         \label{c-speed}
     \end{subfigure}
     \hfill
     \begin{subfigure}[b]{0.3\textwidth}
         \centering
         \includegraphics[height=\textwidth]{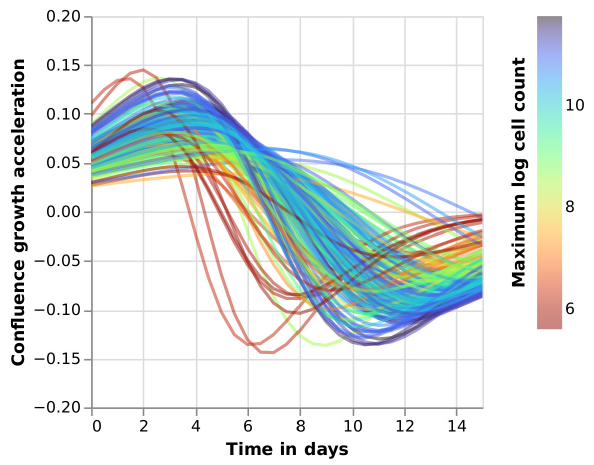}
         \caption{Acceleration.}
         \label{n-acc}
     \end{subfigure}
	\hfill
     \begin{subfigure}[b]{0.09\textwidth}
     \end{subfigure}
        \caption{Confluence growth patterns (in pixels): raw vs modelled.}
        \label{experiment-c}
\end{figure}

Figure \ref{experiment-c} shows confluence in pixels, confluence growth speed and confluence growth acceleration, along with what is implied by our model. Sub-figures \ref{c-speed-raw} and \ref{c-speed} show confluence growth speed increasing to its maximum confluence growth rate (i.e. $\omega_c$) at the confluence growth time (i.e. $\gamma_c$), and finally declining as the well approaches full confluency. Sub-figures \ref{n-acc-raw} and \ref{n-acc} imply that cell passage time is observed at when the confluence growth acceleration attains a minimum (i.e where the confluence $\approx 78.9\%$, see corollary \ref{c-c}); however, numerical results imply that this relation may not hold true for red-shifted monoclones.\\

\begin{figure}[h!]
     \centering
     \begin{subfigure}[b]{0.23\textwidth}
         \centering
         \includegraphics[height=\textwidth]{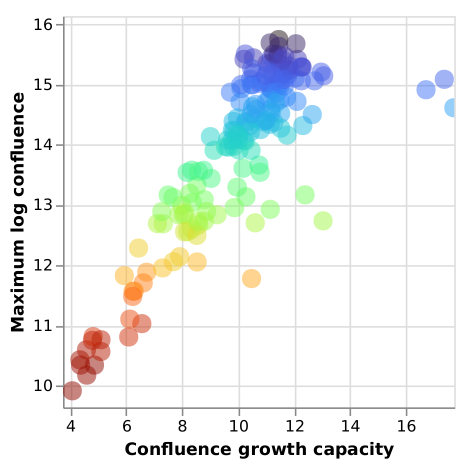}
         \caption{$\beta_c$}
         \label{beta-c}
     \end{subfigure}
     \hfill
     \begin{subfigure}[b]{0.23\textwidth}
         \centering
         \includegraphics[height=\textwidth]{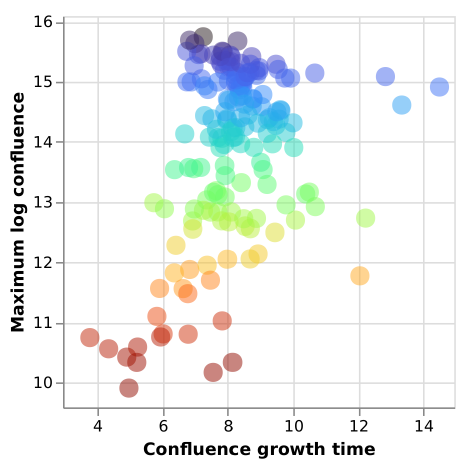}
         \caption{$\gamma_c$}
         \label{gamma-c}
     \end{subfigure}
     \hfill
     \begin{subfigure}[b]{0.23\textwidth}
         \centering
         \includegraphics[height=\textwidth]{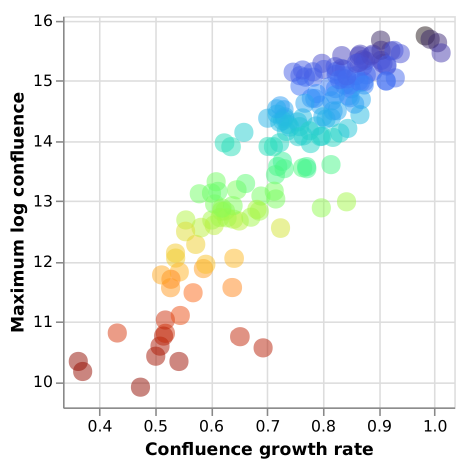}
         \caption{$\omega_c$}
         \label{omega-c}
     \end{subfigure}
	\hfill
     \begin{subfigure}[b]{0.23\textwidth}
         \centering
         \includegraphics[height=\textwidth]{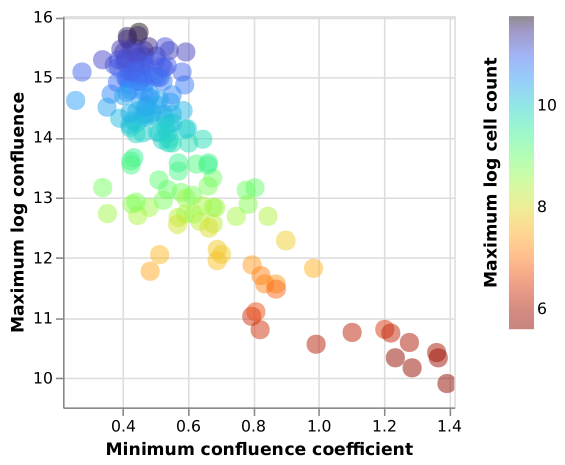}
         \caption{$\alpha_c$}
         \label{alpha-c}
     \end{subfigure}
	\hfill
     \begin{subfigure}[b]{0.07\textwidth}
     \end{subfigure}
        \caption{Confluence growth variables.}
        \label{c-var}
\end{figure}

Figure \ref{c-var} shows the confluence growth variables with respect to maximum of the log confluence, where the confluence is in pixels. Figure \ref{beta-c} implies that the confluence growth capacity is positively correlated with the final log cell count ($\rho=0.844$), figure \ref{omega-n} implies that the confluence growth rate is positively correlated with the final log cell count ($\rho=0.890$) and figure \ref{alpha-n} implies that minimum confluence coefficient is negatively correlated with the final log  cell count ($\rho=-0.827$). Figure \ref{gamma-n}, implies a positive correlation between the confluence growth time and the final log cell count ($\rho=0.384$); however, this correlation becomes negative for monoclones with a final confluence above $\log(15)$, implying a lower confluence growth time results in a greater final cell count for blue-shifted monoclones. Recall that these are the very same correlations that we observed in section \ref{n-plots} with population growth variables. This is a very intuitive result due to the strong interdependence between the cell count and the confluence (discussed in section \ref{sec-d}).

\section{Modelling Cell Area Distribution \label{log-normal}}

In this section, with cell area collected from TEK Optima Research Ltd. DeepInsight\textsuperscript{\textregistered} cell analysis software \cite{tekor}, we perform Kolmogorov-Smirnov analysis to find cell area distributions of our dataset. By cell area, we mean the confluence per cell (not to be confused with the surface area of a cell), where the confluence per cell  is measured in pixels.\\

\begin{figure}[h!]
     \centering
     \begin{subfigure}[b]{0.3\textwidth}
         \centering
         \includegraphics[height=\textwidth]{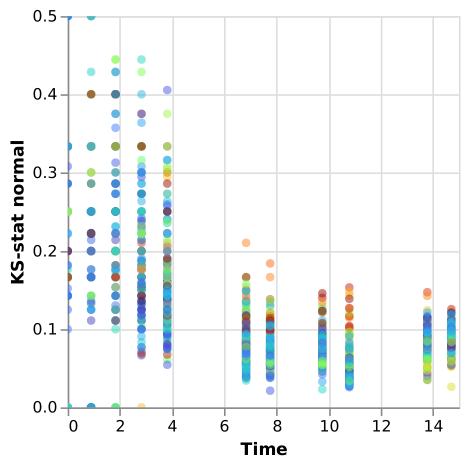}
         \caption{$\mathcal{N}$ error.}
         \label{nn}
     \end{subfigure}
     \hfill
     \begin{subfigure}[b]{0.3\textwidth}
         \centering
         \includegraphics[height=\textwidth]{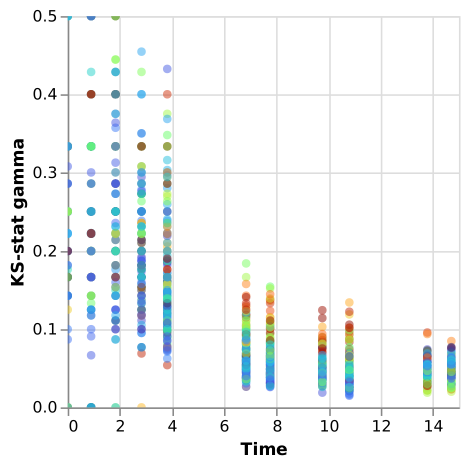}
         \caption{$\Gamma$ error.}
         \label{gg}
     \end{subfigure}
     \hfill
     \begin{subfigure}[b]{0.3\textwidth}
         \centering
         \includegraphics[height=\textwidth]{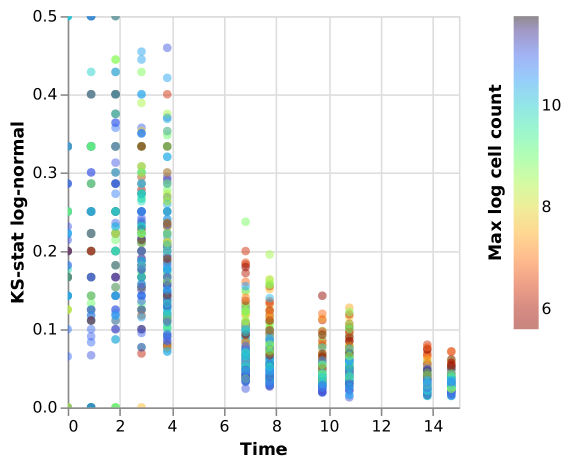}
         \caption{$\log\mathcal{N}$ error.}
         \label{lnln}
     \end{subfigure}
	\hfill
     \begin{subfigure}[b]{0.09\textwidth}
     \end{subfigure}
	\hfill
     \begin{subfigure}[b]{0.3\textwidth}
         \centering
         \includegraphics[height=\textwidth]{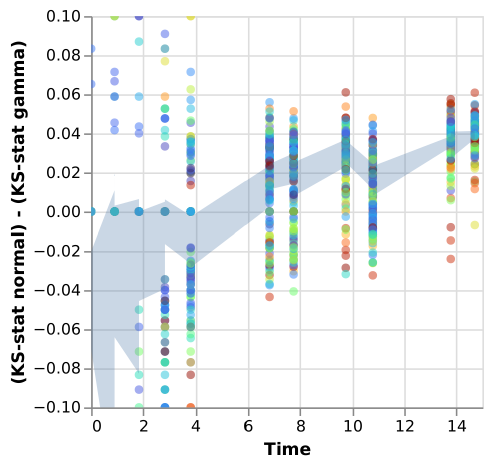}
         \caption{$\mathcal{N}$ error $-$ $\Gamma$ error.}
         \label{ng}
     \end{subfigure}
	\hfill
     \begin{subfigure}[b]{0.3\textwidth}
         \centering
         \includegraphics[height=\textwidth]{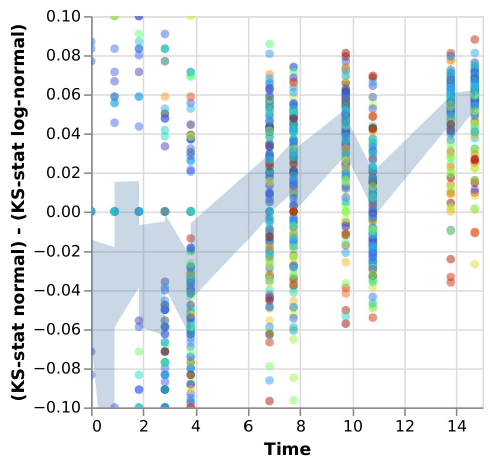}
         \caption{$\mathcal{N}$ error $-$ $\log\mathcal{N}$ error.}
         \label{nln}
     \end{subfigure}
	\hfill
     \begin{subfigure}[b]{0.3\textwidth}
         \centering
         \includegraphics[height=\textwidth]{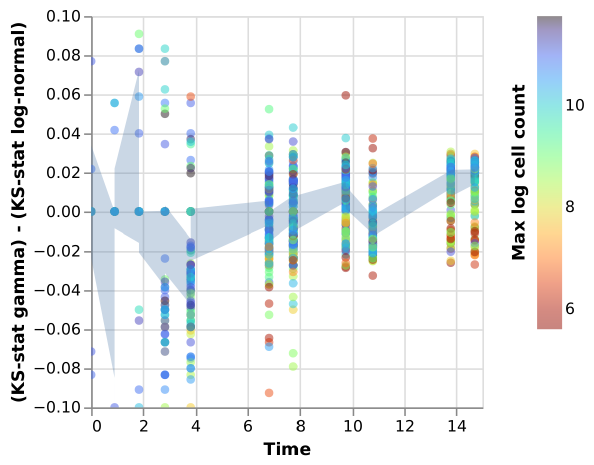}
         \caption{$\Gamma$ error $-$ $\log\mathcal{N}$ error.}
         \label{gln}
     \end{subfigure}
	\hfill
     \begin{subfigure}[b]{0.09\textwidth}
     \end{subfigure}
        \caption{Kolmogorov-Smirnov statistic for normal, gamma and log-normal distributions, and their difference.}
        \label{ks}
\end{figure}

Figure \ref{ks} shows the Kolmogorov-Smirnov statistic (can be interpreted as the error) of cell area distribution with respect to normal ($\mathcal{N}$), gamma ($\Gamma$) and log-normal ($\log\mathcal{N}$) distributions, where the Kolmogorov-Smirnov statistic is calculated with \emph{SciPy} \emph{kstest} \cite{scipy.kstest}. Should one examines figures \ref{nn}, \ref{gg} and \ref{lnln} one may see that as the cell population grows, cell area distribution becomes log-normal. However, should one examine the difference in the Kolmogorov-Smirnov statistics of normal and gamma (see figure \ref{ng}), normal and log-normal (see figure \ref{nln}), and gamma and log-normal (see figure \ref{gln}), one can see that in the first 4 days, area per cell is normally distributed, day 6 to day 11 and for red-shifted samples (i.e. low final cell count samples), area per cell is gamma distributed, and day 8 to day 15 and blue-shifted samples (i.e. high final cell count samples), area per cell is log-normally distributed, where the blue $95\%$ confidence bands are calculated with \emph{Vega-Altair} visualisation software \cite{altair}. These observations leads to the following theorem,

\begin{Theorem}\label{thrm-ks}
Kolmogorov-Smirnov statistic implies that the cell areas of the incipient monoclones are normally distributed, cell areas of the sparse cell population are gamma distributed and cell areas of the dense colony population are log-normally distributed.
\end{Theorem}

\begin{figure}[h!]
     \centering
     \begin{subfigure}[b]{0.40\textwidth}
         \centering
         \includegraphics[height=\textwidth]{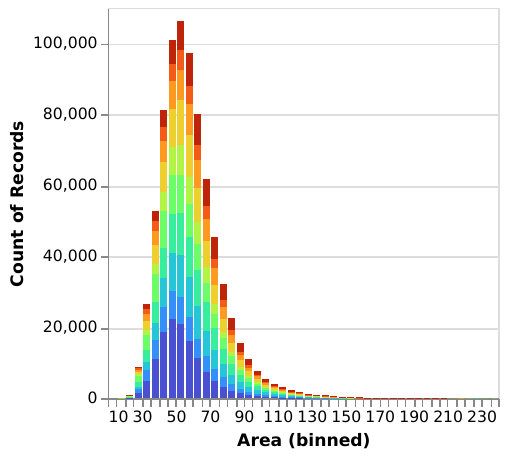}
         \caption{Raw area.}
         \label{a-raw}
     \end{subfigure}
     \hfill
     \begin{subfigure}[b]{0.40\textwidth}
         \centering
         \includegraphics[height=\textwidth]{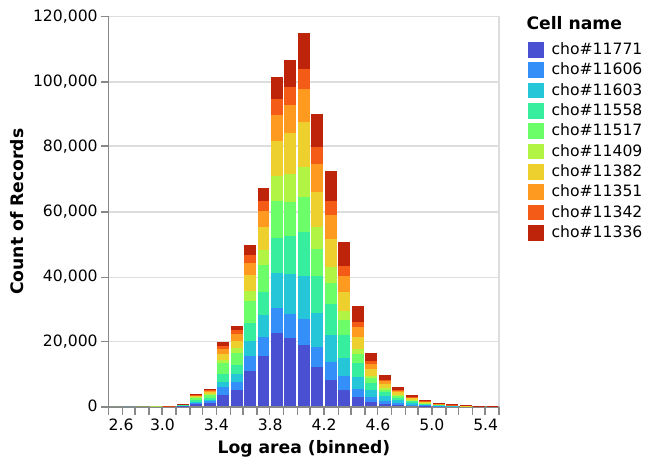}
         \caption{Raw log area.}
         \label{log-a-raw}
     \end{subfigure}
	\hfill
     \begin{subfigure}[b]{0.19\textwidth}
     \end{subfigure}
     \hfill
     \begin{subfigure}[b]{0.40\textwidth}
         \centering
         \includegraphics[height=\textwidth]{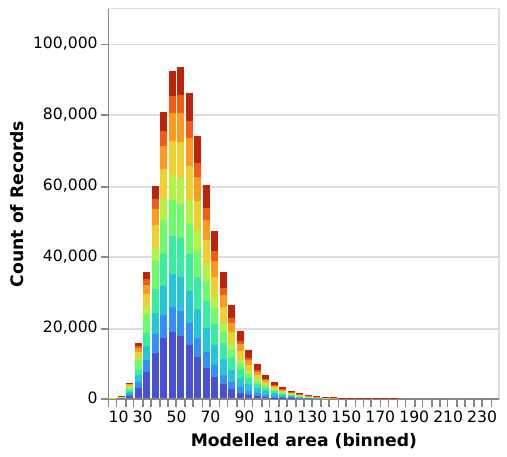}
         \caption{Area.}
         \label{a}
     \end{subfigure}
     \hfill
     \begin{subfigure}[b]{0.4\textwidth}
         \centering
         \includegraphics[height=\textwidth]{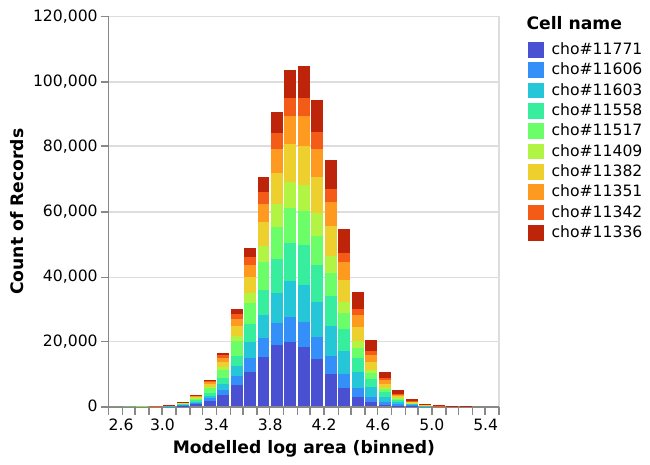}
         \caption{Log area.}
         \label{log-a}
     \end{subfigure}
	\hfill
     \begin{subfigure}[b]{0.19\textwidth}
     \end{subfigure}
        \caption{Cell area distributions of 10 CHO monoclones on day 14: raw vs log-normally modelled.}
        \label{histo}
\end{figure}

Our goal is to study the behaviour of large populations, and as we established that the cell area of large populations are log-normally distributed in theorem \ref{thrm-ks}, herein we model cell areas to be normally distributed. Figure \ref{histo} shows the cell area distributions of 10 CHO monoclones on day 14, where figures \ref{a-raw} and \ref{log-a-raw} show the observed cell area distributions and figures \ref{a} and \ref{log-a} show the log-normally modelled cell area distributions.\\

\begin{table}[h!]
\begin{tabular}{ |l||r|r|r|r|r|r|r|}
\toprule
\multirow{1}{*}{Cell name } & \multicolumn{1}{ |c |}{ Cell} &  \multicolumn{2}{ |c |}{KS-stat (error)} & \multicolumn{4}{ |c |}{ Log area} \\
\cmidrule{3-8}
 & \multicolumn{1}{ |c |}{ count} &\multicolumn{1}{ |c |}{ $\Gamma$}& \multicolumn{1}{ |c |}{ $\log\mathcal{N}$ } & \multicolumn{1}{ |c |}{ Mean} &	\multicolumn{1}{ |c |}{ STD}&	\multicolumn{1}{ |c |}{ Max} &	\multicolumn{1}{ |c |}{ Min}\\
\midrule
$cho\#11771$ 	&	129,508	&0.06447&0.04343	&	3.941	&	0.2586	&	5.583	&	2.197\\
$cho\#11606$ 	&	61,533	&0.04342&0.02480	&	3.992	&	0.3148	&	5.656	&	2.197\\
$cho\#11603$ 	&	86,873	&0.07277&0.04712	&	4.051	&	0.3013	&	5.710	&	2.708\\
$cho\#11558$ 	&	91,368	&0.03756&0.01874&	4.030	&	0.2937	&	5.704	&	2.773\\
$cho\#11517$ 	&	82,960	&0.04825&0.02586&	3.961	&	0.3080	&	5.618	&	2.303\\
$cho\#11409$ 	&	53,683	&0.05849&0.03592	&	4.008	&	0.2588	&	5.525	&	2.944\\
$cho\#11382$ 	&	87,778	&0.06624&0.04036	&	4.037	&	0.2900	&	5.690	&	2.197\\
$cho\#11351$ 	&	67,354	&0.04604&0.02442	&	4.045	&	0.2914	&	5.697	&	2.890\\
$cho\#11342$ 	&	41,011	&0.03619&0.01822	&	4.023	&	0.2846	&	5.468	&	2.708\\
$cho\#11336$ 	&	73,343	&0.05153&0.02807	&	4.115	&	0.2957	&	5.727	&	2.197\\
\bottomrule
\end{tabular}
\caption{Cell log area mean, standard deviation, maximum and minimum, the cell count, and the Kolmogorov-Smirnov statistic with respect to gamma and log-normal distributions of 10 CHO monoclones on day 14.}
\label{tab0}
\end{table}

Table \ref{tab0} shows the raw log cell areas statistics for the 10 CHO monoclones that is used to generate figures \ref{a} and \ref{log-a}, where cell area data in table \ref{tab0} rounded to 4 significant figures. The reader may find that log-normal distribution is a better fit than gamma distribution for all the 10 samples. Also, should one examine the data more carefully, one finds that both mean ($\rho=-0.360$) and standard deviation ($\rho=-0.213$) of  log cell area are negatively correlated with the cell count. This motivates us to study how the cell areas grow as a function of time, in the subsequent sections.

\section{Mean Log Cell Area}

Recall that in section \ref{log-normal}, we demonstrate that cell area is log-normally distributed. Thus, following similar reasonings as in section (\ref{cell-count}), to model the mean log cell area as a function of time, i.e. $y_r(x) = \mathbb{E}(\log(\text{cell area})) $, we arrive at the following hypothesis,

\begin{hypothesis} \label{initial-radial-speed}
Let radial growth time be the time where we observe maximum rate of change in cell area. Then, the cell area of the incipient cell is approximately equals to the mean cell area at the radial growth time, and the growth speed of the incipient cell is approximately equals to the magnitude of radial growth speed at the radial growth time.
\end{hypothesis}

Thus, we arrive at the following theorem,

\begin{Theorem}[Radial Growth Model]\label{cell-area}
Assuming hypothesis \ref{initial-radial-speed}, mean log cell area growth in a confined space can be modelled by the following equation,
\begin{align}
 y_r(x) = \beta_r [\alpha_r + 1  - \boldsymbol{\sigma}(\theta_r (x - \gamma_r)) + \epsilon_r\log (\boldsymbol{\sigma}(\rho_r \theta_r x))] , \label{radial-eqn}
\end{align}
where $y_r(x)$ is the mean natural log cell area, $x$ is the time, $\boldsymbol{\sigma}(\cdot)$ is the sigmoid function (definition \ref{sigmoid-diff}),
\begin{align*}
\theta_r & = 4 \frac{\omega_r}{\beta_r},\\
\rho_r &= \frac{\delta_r}{2\epsilon_r} ,
\end{align*}
$\beta_r$ is the radial growth capacity, $\omega_r$ is radial growth rate, $\gamma_r$ is the radial growth time, 
\begin{align*}
\epsilon_r =\frac{1 - 2\boldsymbol{\sigma}(-\theta_r \gamma_r)} {2\log(2) }
\end{align*}
is the incipient radial growth capacity coefficient,
\begin{align*}
\delta_r = 1 + 4\boldsymbol{\sigma}(-\theta_r \gamma_r) (1 - \boldsymbol{\sigma}(-\theta_r \gamma_r)) 
\end{align*}
is the incipient radial growth rate coefficient, 
\begin{align*}
\alpha_r  =\mathbb{E}\left[ \frac{1}{\beta_r } y_r(x) + \boldsymbol{\sigma}(\theta_r (x - \gamma_r)) - \epsilon_r\log (\boldsymbol{\sigma}(\rho_r \theta_r x)) \mid (\beta_r, \gamma_r, \omega_r) \right] - 1
\end{align*}
is the minimum radial coefficient, $\mathbb{E}(\cdot)$ is the expectation operator, and where $\beta_r$, $\omega_r$ and $\gamma_r$ are the only independent variables of the model, and $x$ is the only independent variable of the dataset.
\end{Theorem}

Interpretation of theorem \ref{cell-area} is as follows. $\beta_r$ defines the capacity for a sparse cell to be large. $\omega_r$ defines maximum rate of change of cell radius, were we observe a reduction in area as a sparse cell become a part of a colony. $\gamma_r$ defines the time where we observe the largest rate of reduction in cell radius. This may indicate the time where most cells start to form dense colonies. $\epsilon_r$ defines capacity of the incipient cell to grow in area, instead of multiplying. $\delta_n$ defines the rate of area growth of the incipient cell. $\alpha_r$ defines the $\log$ of the area of an average cell in a colony, normalised by $\beta_r$. Now, taking first and second order derivatives of the equation \ref{cell-area}, we arrive at the following corollary,

\begin{corollary}
Theorem \ref{cell-area} implies that the radial growth speed and radial growth acceleration can respectively expressed as follows,
\begin{align*}
\frac{d y_r(x)}{dx}  &= \omega_r \left[ - \frac{4}{\theta_r} \boldsymbol{\sigma}^{(1)}(\theta_r (x - \gamma_r) )+ 2\delta_r (1 - \boldsymbol{\sigma}(\rho_r \theta_r x))\right]~\text{and}  \\
\frac{d^2 y_r(x)}{dx^2} &= - \omega_r \theta_r\left[ \frac{4}{(\theta_r)^2} \boldsymbol{\sigma}^{(2)}(\theta_r (x - \gamma_r)) + 2\frac{\delta_r}{\theta_r} \boldsymbol{\sigma}^{(1)}(\rho_r \theta_r x)  \right],
\end{align*}
where $\boldsymbol{\sigma}^{(1)}(\cdot)$ and  $\boldsymbol{\sigma}^{(2)}(\cdot)$ are the first and second order derivatives of the sigmoid function (definition \ref{sigmoid-diff}).
\end{corollary}

\subsection{Numerical Modelling}

Given a dataset $(x, y_r)$, to find a $(\beta_r, \gamma_r, \omega_r, \alpha_r)$-set, we precent the following algorithm.\\

\emph{Step 1:} Normalise the dataset. First normalise the time data points as $x_{\text{norm}} =\frac{x}{ x_{\text{max}} }$, where $x_{\text{max}} =\max(x)$. Then, using \emph{NumPy} \emph{polyfit} function with \emph{deg}$=1$ \cite{numpy.polyfit}, fit a line of best fit to the mean log cell area data points as $(y_r)_\text{$-$ve} = a + b x_{\text{norm}}$, where $(y_r)_\text{$-$ve} = -y_r$. Now, using the coefficients $a$ and $b$, normalise the log area data points as $(y_r)_{\text{norm}} =\frac{(y_r)_\text{$-$ve} - (y_r)_{\text{min}}}{ (y_r)_{\text{max}} }$, where $(y_r)_{\text{min}} = a $, $(y_c)_{\text{max}} = b +  \text{RMSE}$.\\

\emph{Step 2:} Find lower and upper bounds for the parameters. Using  \emph{SciPy} \emph{curve\_fit} function with \emph{maxfev}$=10,000$ \cite{scipy.curvefit}, fit the normalised data to the following equation,
\begin{align*}
(y_r)_{\text{norm}} = b_r \boldsymbol{\sigma}(4 d_r (x_{\text{norm}} - c_r)),
\end{align*}
where $b_r$, $c_r$ and $d_r$ are bounded below by $0$ and above by $2$.\\

\emph{Step 3:} Find normalised radial growth parameters. Using the bounds $0<\beta_r^0 <b_r$, $c_r< \gamma_r^0 < 2$ and $0 < \omega^0_r< b_r d_r$, and using \emph{SciPy} \emph{curve\_fit} function with \emph{maxfev}$=10,000$ \cite{scipy.curvefit}, fit the data to the following equation,
\begin{align*}
(y_r)_{\text{norm}} =   \beta^0_r \left[ \boldsymbol{\sigma}\left(\theta^0_r(x_{\text{norm}} - \gamma^0_r)\right) - \boldsymbol{\sigma}\left(-\frac12 \theta^0_r \gamma^0_r\right)  - \epsilon^0_r\log \left(\frac{\boldsymbol{\sigma}\left(\rho^0_r \theta^0_r x_{\text{norm}}\right)}{\boldsymbol{\sigma}\left(\frac12 \rho^0_r \theta^0_r \gamma^0_r \right)} \right) \right],
\end{align*}
where  $\epsilon^0_r$, $\theta^0_r$ and $\rho^0_r$  dependent variables of $\beta_r^0$, $ \gamma_r^0$ and $\omega^0_r$, which are defined in theorem \ref{cell-area}. With $\beta_r^0$, $ \gamma_r^0$ and $\omega^0_r$, find $\alpha_r^0$  as follows,
\begin{align*}
\alpha^0_r  =\mathbb{E}\left[ (y_r)_{\text{norm}} -  \beta_r^0\left[\boldsymbol{\sigma}\big(\theta_r^0 \big(x_{\text{norm}} - \gamma_r^0\big)\big) - \epsilon^0_r\log \big(\boldsymbol{\sigma}\big(\rho^0_r \theta^0_r x\big)\big) \right] \right] +  \frac{(y_r)_{\text{min}}}{(y_r)_{\text{max}}} + 1.
\end{align*}

\emph{Step 4:} Unnormalise the radial growth parameters as $\alpha_r  = -\frac{1}{\beta_r^0}  \alpha^0_r$, $\beta_r  = (y_r)_{\text{max}} \beta_r^0 $,  $\gamma_r  = x_{\text{max}} \gamma_r^0 $ and  $\omega_r  = \frac{(y_r)_{\text{max}} }{x_{\text{max}}} \omega_r^0 $.\\

For a sample dataset, along with a working algorithm, please see the link in the footnote\footnote{Radial growth model  \emph{Colab} notebook with a sample dataset, a working algorithm and plots:  
\url{https://drive.google.com/file/d/1kKGsSgs8yjY9GWsdEZ8TCCcfDhVkKduP/view?usp=drive_link}} for a \emph{Colab} notebook.

\subsection{Experimental Results}

In this section, we fit the our radial growth model  (theorem \ref{cell-area}) to \emph{CHO2023} dataset.\\

\begin{figure}[h!]
     \centering
     \begin{subfigure}[b]{0.3\textwidth}
         \centering
         \includegraphics[height=\textwidth]{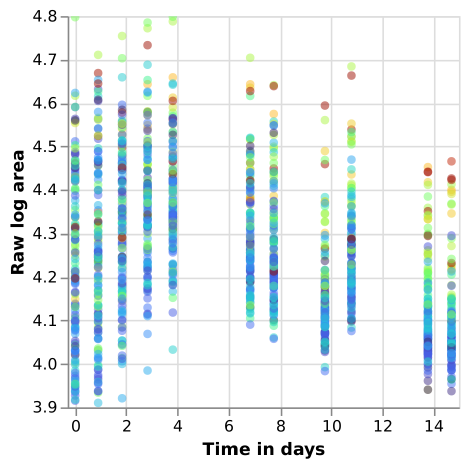}
         \caption{Raw mean log area.}
         \label{r-capacity-raw}
     \end{subfigure}
     \hfill
     \begin{subfigure}[b]{0.3\textwidth}
         \centering
         \includegraphics[height=\textwidth]{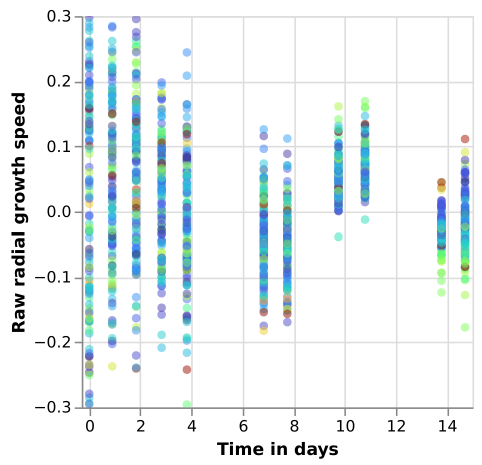}
         \caption{Raw speed.}
         \label{r-speed-raw}
     \end{subfigure}
     \hfill
     \begin{subfigure}[b]{0.3\textwidth}
         \centering
         \includegraphics[height=\textwidth]{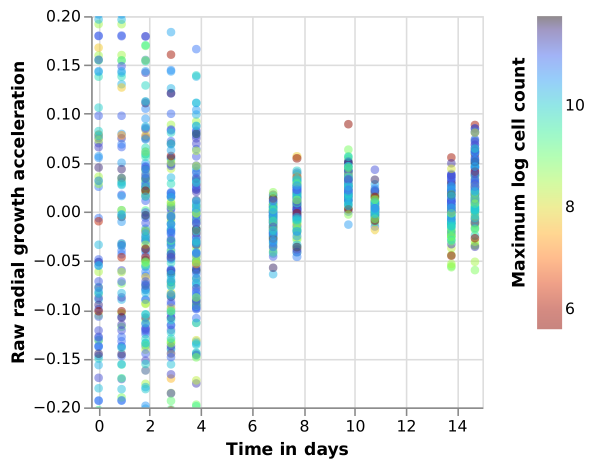}
         \caption{Raw acceleration.}
     \end{subfigure}
	\hfill
     \begin{subfigure}[b]{0.09\textwidth}
     \end{subfigure}
     \hfill
     \begin{subfigure}[b]{0.3\textwidth}
         \centering
         \includegraphics[height=\textwidth]{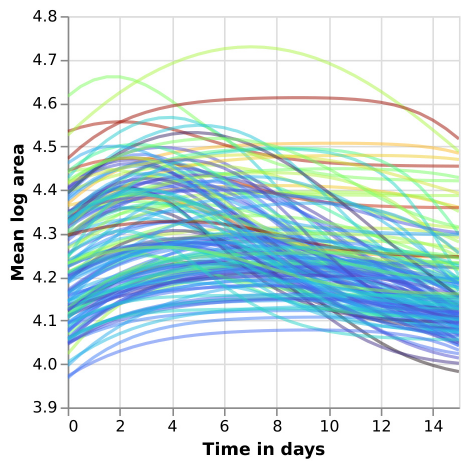}
         \caption{Mean log area.}
         \label{r-capacity}
     \end{subfigure}
     \hfill
     \begin{subfigure}[b]{0.3\textwidth}
         \centering
         \includegraphics[height=\textwidth]{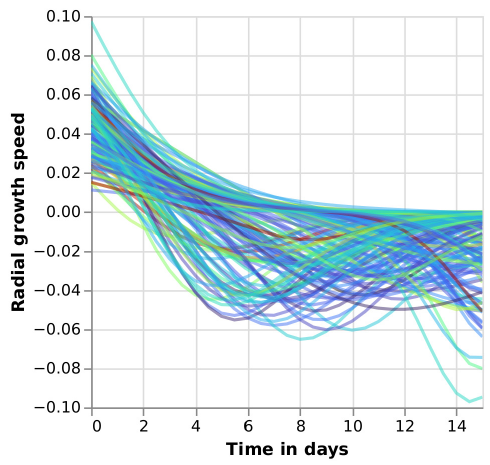}
         \caption{Speed.}
         \label{r-speed}
     \end{subfigure}
     \hfill
     \begin{subfigure}[b]{0.3\textwidth}
         \centering
         \includegraphics[height=\textwidth]{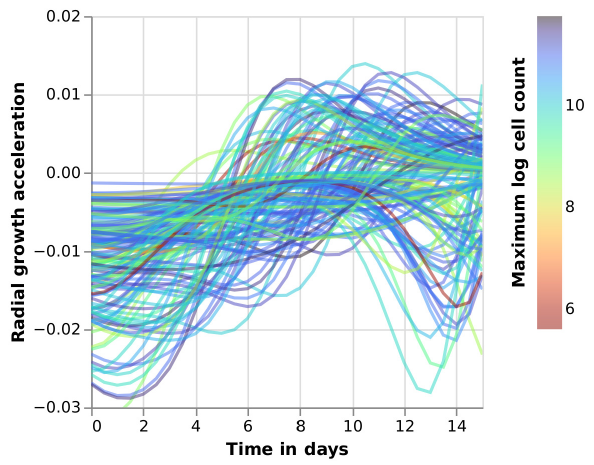}
         \caption{Acceleration.}
     \end{subfigure}
	\hfill
     \begin{subfigure}[b]{0.09\textwidth}
     \end{subfigure}
        \caption{Radial growth patterns: raw vs modelled.}
        \label{experiment-r}
\end{figure}

Figure \ref{experiment-r} shows mean log cell area in pixels, radial growth speed and radial growth acceleration, along with what is implied by our model. Sub-figures \ref{r-capacity-raw} and \ref{r-capacity} show an increase in area per cell in the incipient population (i.e. cells expanding, instead of multiplying), and then a decrease in the area per cell as the population matures (i.e. as the population starts to multiply rapidly and form colonies, area per cell decreases). Sub-figures \ref{r-speed-raw} and \ref{r-speed} radial growth speed declining, reaching its minimum radial growth rate (i.e. $-\omega_r$) at the radial growth time (i.e. $\gamma_r$), and finally approaching zero as time increases.\\

\begin{figure}[h!]
     \centering
     \begin{subfigure}[b]{0.30\textwidth}
         \centering
         \includegraphics[height=\textwidth]{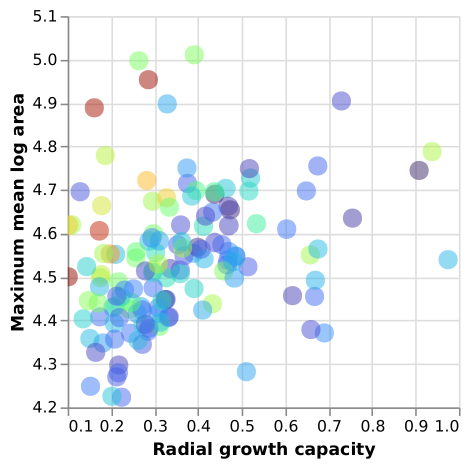}
         \caption{$\beta_r$}
         \label{beta-r}
     \end{subfigure}
     \hfill
     \begin{subfigure}[b]{0.30\textwidth}
         \centering
         \includegraphics[height=\textwidth]{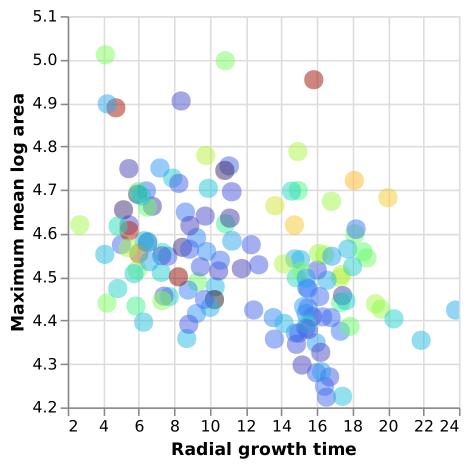}
         \caption{$\gamma_r$}
         \label{gamma-r}
     \end{subfigure}
     \hfill
     \begin{subfigure}[b]{0.30\textwidth}
         \centering
         \includegraphics[height=\textwidth]{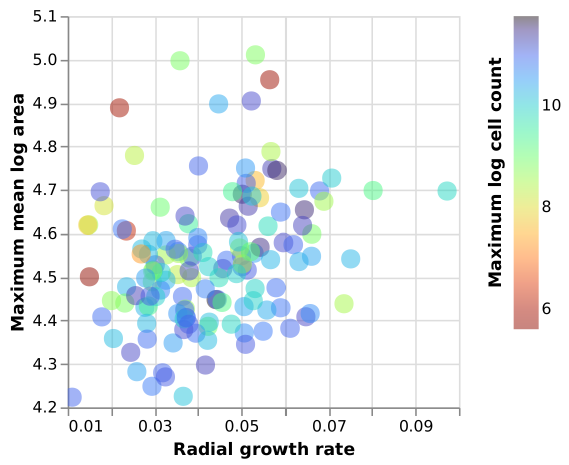}
         \caption{$\omega_r$}
         \label{omega-r}
     \end{subfigure}
	\hfill
     \begin{subfigure}[b]{0.09\textwidth}
     \end{subfigure}
        \caption{Radial growth variables.}
        \label{r-var}
\end{figure}

Figure \ref{r-var} shows the radial growth variables with respect to maximum of the mean log cell area, where cell area is in pixels. Figures \ref{gamma-r} and \ref{omega-r} imply no obvious relationship between   radial growth time ($\rho = -0.031$) and radial growth rate ($\rho = 0.175$), and the final log  cell count. However, figure \ref{beta-r} implies radial growth capacity is positively correlated with final log cell count ($\rho = 0.302$), implying that sparse cells that are large has the capacity to reach a large population.\\

\begin{figure}[h!]
     \centering
     \begin{subfigure}[b]{0.30\textwidth}
         \centering
         \includegraphics[height=\textwidth]{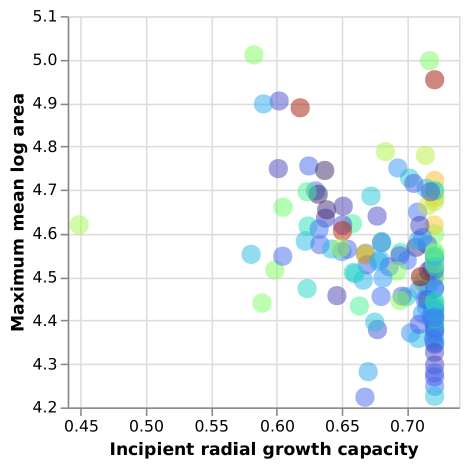}
         \caption{$\epsilon_r$}
         \label{epsilon-r}
     \end{subfigure}
     \hfill
     \begin{subfigure}[b]{0.30\textwidth}
         \centering
         \includegraphics[height=\textwidth]{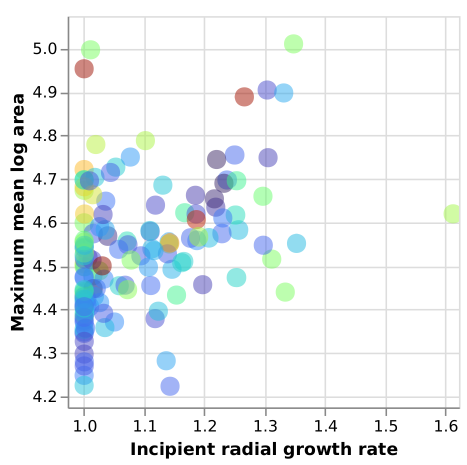}
         \caption{$\delta_r$}
         \label{delta-r}
     \end{subfigure}
     \hfill
     \begin{subfigure}[b]{0.30\textwidth}
         \centering
         \includegraphics[height=\textwidth]{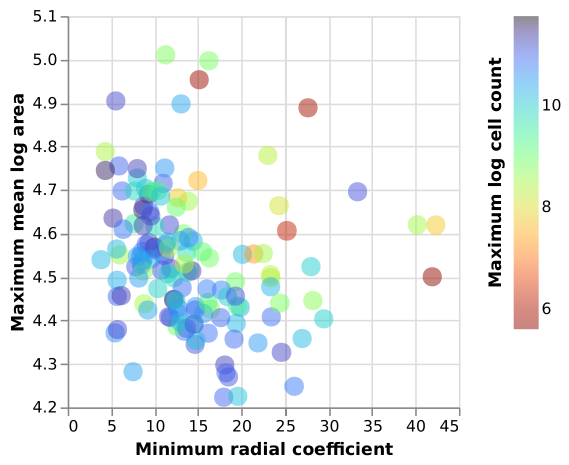}
         \caption{$\alpha_r$}
         \label{alpha-r}
     \end{subfigure}
	\hfill
     \begin{subfigure}[b]{0.09\textwidth}
     \end{subfigure}
        \caption{Incipient radial growth coefficients.}
        \label{r-var-2}
\end{figure}

Figure \ref{r-var-2} shows the incipient radial coefficients with respect to maximum of the log cell count. Figures \ref{epsilon-r} and \ref{delta-r} imply no obvious relationship between  incipient radial growth capacity coefficient ($\rho = 0.022$) and incipient radial growth rate coefficient ($\rho = -0.013$), and final log  cell count. However, figure \ref{beta-r} implies radial growth capacity is negatively correlated with final log cell count ($\rho = - 0.438$), implying that cells that can from very compact colonies has the capacity to reach a large population.

\section{Standard Deviation of the Log Cell Area \label{sec-d}}

Following similar reasonings as in section (\ref{cell-count}), to model the standard deviation of the log cell area  as a function of time, $y_d(x) = \mathbb{S}\text{td}(\log(\text{cell area})) $ where $\mathbb{S}\text{td}(\cdot)$ is the \emph{standard deviation operator}, we arrive at the following theorem,

\begin{Theorem}[Deviation Growth Model]\label{cell-d-thrm}
Standard deviation of the log cell area growth in a confined space can be modelled by the following equation,
\begin{align}
y_d(x) = \beta_d \left[\alpha_d + 1-  \frac{4}{\theta_d}\boldsymbol{\sigma}^{(1)}(\theta_d (x - \gamma_d)) \right] , \label{cell-d}
\end{align}
where $y_d(x)$ is the standard deviation of the log cell area, $x$ is the time,  $\boldsymbol{\sigma}^{(1)}(\cdot)$ is the first order derivative of the  sigmoid function (definition \ref{sigmoid-diff})
\begin{align*}
\theta_d = \frac{3}{2} \sqrt{3} \frac{\omega_d}{\beta_d},
\end{align*}
$\beta_s$ is the deviation growth capacity, $\omega_d$ is deviation growth rate, $\gamma_d$ is the deviation growth time, 
\begin{align*}
\alpha_d  =\mathbb{E}\left[ \frac{1}{\beta_d} y_d(x) + \frac{4}{\theta_d}\boldsymbol{\sigma}^{(1)}(\theta_d (x - \gamma_d)) \mid (\beta_d, \gamma_d, \omega_d) \right] -1
\end{align*}
is the minimum deviation coefficient, $\mathbb{E}(\cdot)$ is the expectation operator, and where $\beta_d$, $\omega_d$ and $\gamma_d$ are the only independent variables of the model, and $x$ is the only independent variable of the dataset.
\end{Theorem}

Interpretation of theorem \ref{cell-d-thrm} is as follows. $\beta_d$ defines the capacity for a cell to deviate from its average cell area. A larger capacity results in more varied cell areas, and a smaller capacity results in more uniform cell areas. $\omega_d$ defines maximum rate of change, were we observe a minimum in the standard deviation of the cell area. $\gamma_d$ defines the time where we observe the minimum standard deviation, i.e. time where the cells are more uniform. $\alpha_d$ defines the minimum log area standard deviation, normalised by $\beta_d$. Now, taking first and second order derivatives of the equation \ref{cell-d} and observing the definition of the mean of a log-normal distribution, we arrive at the following corollary,

\begin{corollary}
Theorem \ref{cell-d-thrm} implies that the deviation growth speed and deviation growth acceleration can respectively expressed as follows,
\begin{align*}
\frac{d y_d(x)}{dx}  &= - \omega_d \left[ \frac{6\sqrt{3}}{(\theta_d)^2} \boldsymbol{\sigma}^{(2)}(\theta_d (x - \gamma_d)) \right]~\text{and}  \\
\frac{d^2 y_d(x)}{dx^2} &= - \omega_d \theta_d\left[ \frac{6\sqrt{3}}{(\theta_d)^3} \boldsymbol{\sigma}^{(3)}(\theta_d (x - \gamma_d))  \right],
\end{align*}
where $\boldsymbol{\sigma}^{(2)}(\cdot)$ and  $\boldsymbol{\sigma}^{(3)}(\cdot)$ are the  second the third order derivatives of the sigmoid function (definition \ref{sigmoid-diff}). Also, the confluence, the cell count, the mean cell area and the standard deviation of the cell area are related by the following equation,
\begin{align}
\log(y_c(x)) = \log(y_n(x)) + y_r(x) + \frac12 (y_d(x))^2. \label{inter}
\end{align}
\end{corollary}

\begin{proof} [Proof of equation (\ref{inter})]
By definition, we have the following relation,
\begin{align*}
\text{confluence} = \text{cell count} \times \text{mean cell area}.
\end{align*}
Now, taking the $\log$ of the above relation and noting the mean of a log-normal distribution \cite{crow2018lognormal}, we find the following equation,
\begin{align*}
\log(\text{confluence}) = \log(\text{cell count}) +  \mathbb{E}(\log(\text{cell area})) + \frac12 \mathbb{S}\text{td}(\log(\text{cell area}))^2  .
\end{align*}
Thus, equations (\ref{cell-growth-model}), (\ref{cell-confluence-eqn}), (\ref{radial-eqn}) and (\ref{cell-d}) imply the following equation,
\begin{align*}
c_0 + \beta_c\boldsymbol{\sigma}(\theta_c (x - \gamma_c)) = ~& \beta_n [\boldsymbol{\sigma}(\theta_n (x - \gamma_n)) + \epsilon_n\log (\boldsymbol{\sigma}(\rho_n \theta_n x))]  \\
& -  \beta_r [ \boldsymbol{\sigma}(\theta_r (x - \gamma_r)) - \epsilon_r\log (\boldsymbol{\sigma}(\rho_r \theta_r x))] \\
& +  \frac12 \left(\beta_d \left[\alpha_d + 1 -  \frac{4}{\theta_d}\boldsymbol{\sigma}^{(1)}(\theta_d (x - \gamma_d)) \right]\right) ^2 ,
\end{align*}
where $c_0$ is a positive constant. Rearranging the above equation, we find the following equation,
\begin{align*}
c_0 + c_1\boldsymbol{\sigma}(c_3 (x - c_5)) -  c_2\log (\boldsymbol{\sigma}(c_4 x)) \approx  \left( \alpha_d + 1-  \frac{4}{\theta_d}\boldsymbol{\sigma}^{(1)}(\theta_d (x - \gamma_d))  \right) ^2 ,
\end{align*}
where $c_0$ to $c_5$ are positive constants. Numerical analysis can show that the above equation holds true in the time interval $[0, \max(x)]$.
\end{proof}

        \begin{table}[h!]
            \footnotesize
                \begin{tabular}[t]{c||c c cc}
				&	$\beta_n$	&	$\beta_c$	&	$\beta_r$	&	$\beta_d$\\
			\hline
			\hline
			$\beta_n$\!\!\! &	1	&	$+$ve\!\!\!	&	$+$ve\!\!\!	&	$-$ve\!\!\! \\ 
			$\beta_c$\!\!\! &	$+$ve\!\!\!	&	1	&	$-$ve\!\!\!	&	$+$ve\!\!\! \\ 
			$\beta_r$\!\!\! &		$+$ve\!\!\!	&	$-$ve\!\!\!	&	1	&	$+$ve\!\!\! \\
			$\beta_d$\!\!\! &	$-$ve\!\!\!	&	$+$ve\!\!\!	&	$+$ve\!\!\!	&	1\\
                \end{tabular}
                \hfill
                \begin{tabular}[t]{c||c c cc}
				&	$\gamma_n$	&	$\gamma_c$	&	$\gamma_r$	&	$\gamma_d$\\
			\hline
			\hline
			$\gamma_n$\!\!\! &	1	&	$+$ve\!\!\!	&	$+$ve\!\!\!	&	$-$ve\!\!\!	\\ 
			$\gamma_c$\!\!\! &	$+$ve\!\!\!	&	1	&	$-$ve\!\!\!	&	$+$ve\!\!\!	\\ 
			$\gamma_r$\!\!\! &		$+$ve\!\!\!	&	$-$ve\!\!\!	&	1	&	$+$ve\!\!\!	\\
			$\gamma_d$\!\!\! &	$-$ve\!\!\!	&	$+$ve\!\!\!	&	$+$ve\!\!\!	&	1\\
                \end{tabular}
                \hfill
                \begin{tabular}[t]{c||c c cc}
				&	$\omega_n$	&	$\omega_c$	&	$\omega_r$	&	$\omega_d$\\
			\hline
			\hline
			$\omega_n$\!\!\! &	1	&	$+$ve\!\!\!	&	$+$ve\!\!\!	&	$-$ve\!\!\!	\\ 
			$\omega_c$\!\!\! &		$+$ve\!\!\!	&	1	&	$-$ve\!\!\!	&	$+$ve\!\!\!	\\ 
			$\omega_r$\!\!\! &		$+$ve\!\!\!	&	$-$ve\!\!\!	&	1	&	$+$ve\!\!\!	\\
			$\omega_d$\!\!\! &	$-$ve\!\!\!	&	$+$ve\!\!\!	&	$+$ve\!\!\!	&	1\\
                \end{tabular}
                \caption{Growth variables interdependencies, i.e. the sign of the partial derivatives with respect to growth variables, in the time interval $[\gamma_d, \max(x)]$ (this is not an exhaustive list).}
\label{tab-inter}
            \end{table}

Table \ref{tab-inter} shows some of the most notable interdependencies implied by equation (\ref{inter}). By interdependency, we mean if one variable were to change, then the rest of the variables must change in order to satisfy equation (\ref{inter}). For example, table \ref{tab-inter}  implies that if the population growth capacity were to increase, then the confluence growth capacity must increase to satisfy equation (\ref{inter}), given that all the other variables remain constant.

\subsection{Numerical Modelling}

Given a dataset $(x, y_d)$, to find a $(\beta_d, \gamma_d, \omega_d, \alpha_d)$-set, we precent the following algorithm.\\

\emph{Step 1:} Normalise the dataset. First normalise the time data points as $x_{\text{norm}} =\frac{x}{ x_{\text{max}} }$, where $x_{\text{max}} =\max(x)$. Then, normalise the standard deviation of the  log cell area  data points as $(y_d)_{\text{norm}} =\frac{(y_d)_\text{$-$ve} - (y_d)_{\text{min}}}{ (y_d)_{\text{max}} }$, where $(y_c)_{\text{min}} = \mathbb{E}((y_d)_\text{$-$ve}) -  \mathbb{S}\text{td}((y_d)_\text{$-$ve}) $, $(y_d)_{\text{max}} =  2 \mathbb{S}\text{td}((y_d)_\text{$-$ve})$ and  $(y_d)_\text{$-$ve}= - y_d$.\\

\emph{Step 2:} Find lower and upper bounds for the parameters. Using  \emph{SciPy} \emph{curve\_fit} function with \emph{maxfev}$=10,000$  \cite{scipy.curvefit}, fit the normalised data to the following equation,
\begin{align*}
(y_d)_{\text{norm}} = \frac{8}{9} \sqrt{3} \frac{b_d}{d_d} \boldsymbol{\sigma}^{(1)}\left(\frac{3}{2} \sqrt{3} d_d (x_{\text{norm}} - c_d)\right),
\end{align*}
where $b_d$, $c_d$ and $d_d$ are  bounded below by $0$ and above by $2$.\\

\emph{Step 3:} Find normalised deviation growth parameters. Using the bounds $0<\beta_d^0 <b_d$, $c_d< \gamma_d^0 < 2$ and $0 < \omega^0_d< b_d d_d$, and using \emph{SciPy} \emph{curve\_fit} function with \emph{maxfev}$=10,000$  \cite{scipy.curvefit}, fit the data to the following equation,
\begin{align*}
(y_d)_{\text{norm}} = 4 \frac{\beta^0_d}{\theta^0_d} \big[\boldsymbol{\sigma}^{(1)}\big(\theta^0_d(x_{\text{norm}} - \gamma^0_d)\big) - \boldsymbol{\sigma}^{(1)}\big(-\theta^0_d \gamma^0_d\big) \big],
\end{align*}
where $\theta^0_d = \frac{3}{2} \sqrt{3} \frac{\omega^0_d}{\beta^0_d}$. With $\beta_d^0$, $ \gamma_d^0$ and $\omega^0_d$, find $\alpha_d^0$  as follows,
\begin{align*}
\alpha^0_d  =\mathbb{E}\left[ (y_d)_{\text{norm}} - 4 \frac{\beta^0_d}{\theta^0_d} \boldsymbol{\sigma}^{(1)}\big(\theta_d^0 \big(x_{\text{norm}} - \gamma_d^0\big)\big)\right]  + \frac{(y_d)_{\text{min}}}{(y_d)_{\text{max}}} + 1 .
\end{align*}

\emph{Step 4:} Unnormalise the radial growth parameters as  $\alpha_d = -\frac{1}{\beta_d^0}  \alpha^0_d$, $\beta_d  = (y_d)_{\text{max}} \beta_d^0 $,  $\gamma_d  = x_{\text{max}} \gamma_d^0 $ and  $\omega_d  = \frac{(y_d)_{\text{max}} }{x_{\text{max}}} \omega_d^0 $.\\

For a sample dataset, along with a working algorithm, please see the link in the footnote\footnote{Deviation growth model  \emph{Colab} notebook with a sample dataset, a working algorithm and plots: 
\url{https://drive.google.com/file/d/1HdnuSsFTGBeDvaYjh_4riIx4CPfgmoa3/view?usp=drive_link}} for a \emph{Colab} notebook.

\subsection{Experimental Results}

In this section, we fit the our deviation growth model  (theorem \ref{cell-d-thrm}) to \emph{CHO2023} dataset.\\

\begin{figure}[h!]
     \centering
     \begin{subfigure}[b]{0.3\textwidth}
         \centering
         \includegraphics[height=\textwidth]{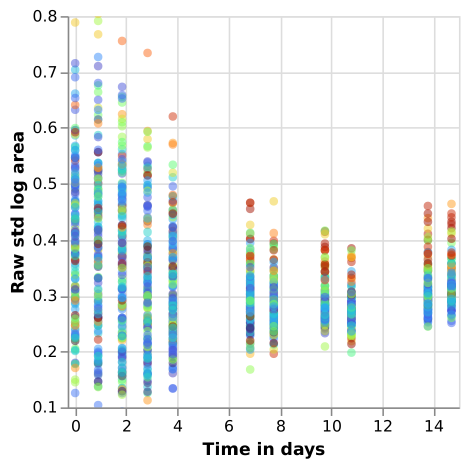}
         \caption{Raw log area Std.}
         \label{d-capacity-raw}
     \end{subfigure}
     \hfill
     \begin{subfigure}[b]{0.3\textwidth}
         \centering
         \includegraphics[height=\textwidth]{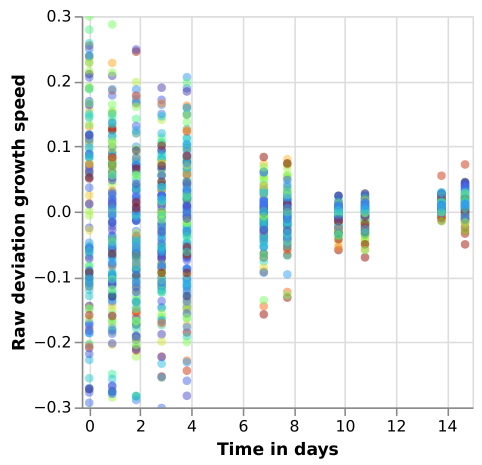}
         \caption{Raw speed.}
     \end{subfigure}
     \hfill
     \begin{subfigure}[b]{0.3\textwidth}
         \centering
         \includegraphics[height=\textwidth]{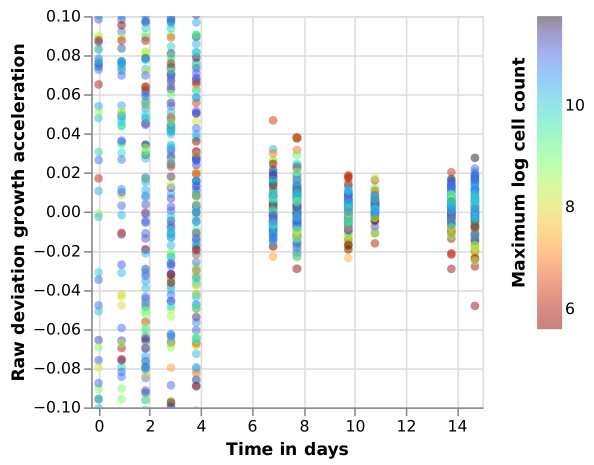}
         \caption{Raw acceleration.}
     \end{subfigure}
	\hfill
     \begin{subfigure}[b]{0.09\textwidth}
     \end{subfigure}
     \hfill
     \begin{subfigure}[b]{0.3\textwidth}
         \centering
         \includegraphics[height=\textwidth]{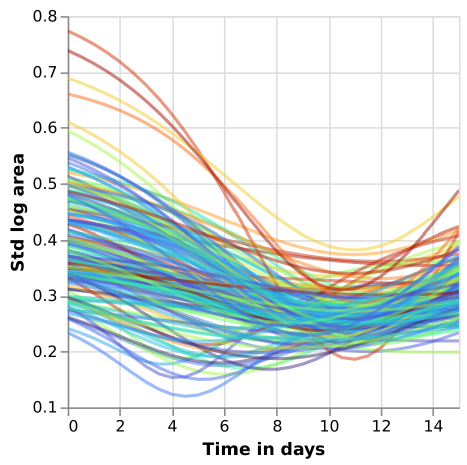}
         \caption{Log area Std.}
         \label{d-capacity}
     \end{subfigure}
     \hfill
     \begin{subfigure}[b]{0.3\textwidth}
         \centering
         \includegraphics[height=\textwidth]{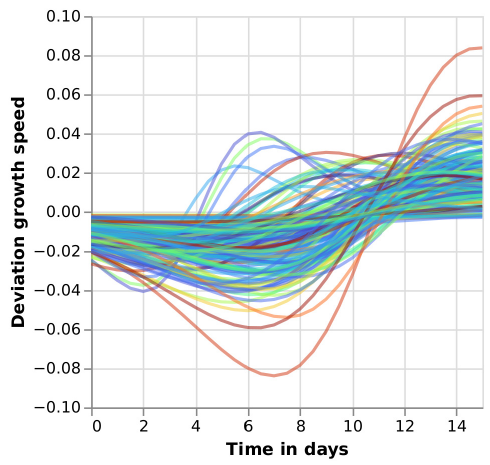}
         \caption{Speed.}
     \end{subfigure}
     \hfill
     \begin{subfigure}[b]{0.3\textwidth}
         \centering
         \includegraphics[height=\textwidth]{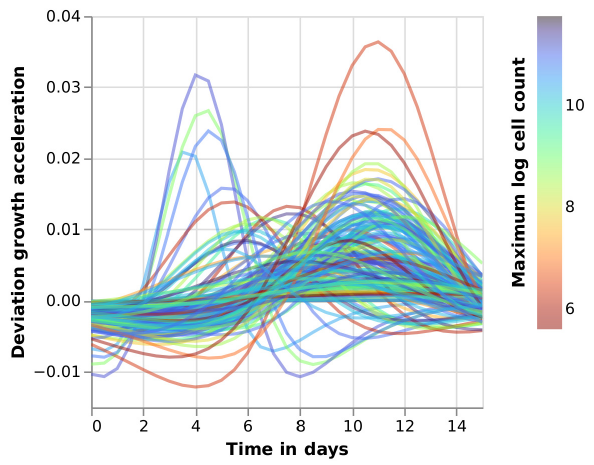}
         \caption{Acceleration.}
     \end{subfigure}
	\hfill
     \begin{subfigure}[b]{0.09\textwidth}
     \end{subfigure}
        \caption{Deviation growth patterns: raw vs modelled.}
        \label{experiment-d}
\end{figure}

Figure \ref{experiment-d} shows the standard deviation of the log cell area in pixels, radial growth speed and radial growth acceleration, along with what is implied by our model. Sub-figures \ref{d-capacity-raw} and \ref{d-capacity} show that the standard deviation of the log cell area declines, reaching a minimum at the deviation growth time (i.e. $\gamma_d$), and increasing as the population grows. This implies that initially, cell sizes can vary significantly, and as cells tends to multiply, cells sizes becomes more uniform, and finally as cells form dense colonies, cell sizes can again vary significantly.\\

\begin{figure}[!h]
     \centering
     \begin{subfigure}[b]{0.23\textwidth}
         \centering
         \includegraphics[height=\textwidth]{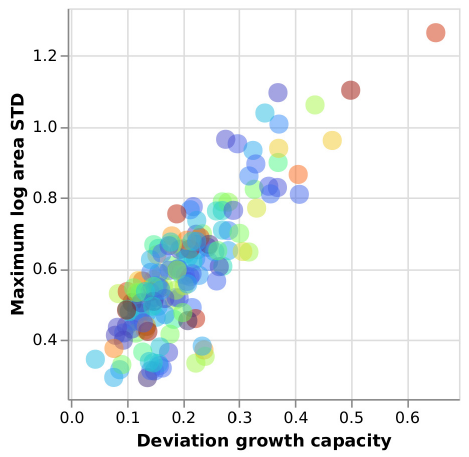}
         \caption{$\beta_d$}
         \label{beta-d}
     \end{subfigure}
     \hfill
     \begin{subfigure}[b]{0.23\textwidth}
         \centering
         \includegraphics[height=\textwidth]{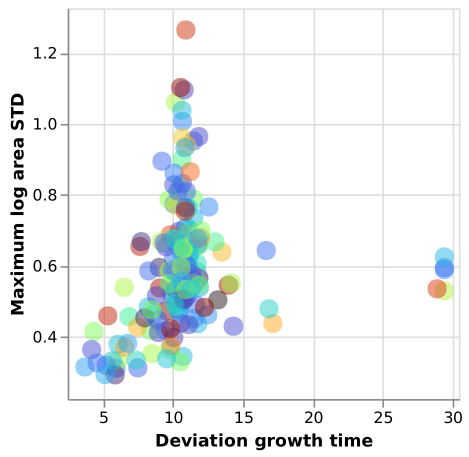}
         \caption{$\gamma_d$}
         \label{gamma-d}
     \end{subfigure}
     \hfill
     \begin{subfigure}[b]{0.23\textwidth}
         \centering
         \includegraphics[height=\textwidth]{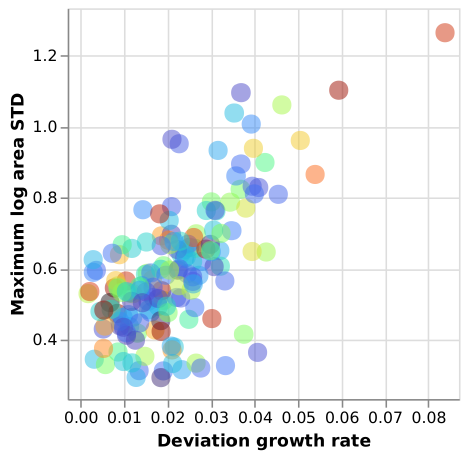}
         \caption{$\omega_d$}
         \label{omega-d}
     \end{subfigure}
	\hfill
     \begin{subfigure}[b]{0.23\textwidth}
         \centering
         \includegraphics[height=\textwidth]{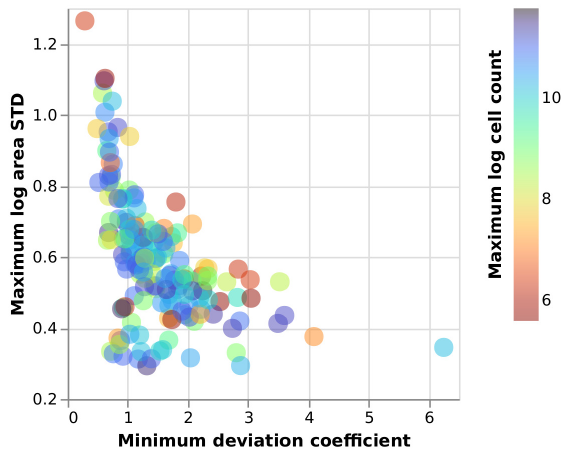}
         \caption{$\alpha_d$}
         \label{alpha-d}
     \end{subfigure}
	\hfill
     \begin{subfigure}[b]{0.07\textwidth}
     \end{subfigure}
        \caption{Deviation growth variables.}
        \label{d-var}
\end{figure}

Figure \ref{d-var} shows the radial growth variables with respect to maximum of the standard deviation of the log cell area in pixels. Figures \ref{beta-d}, \ref{omega-d}, \ref{gamma-d} and \ref{alpha-d} imply no obvious relationship between deviation growth capacity ($\rho=-0.095$), deviation growth rate ($\rho=-0.075$), deviation growth time ($\rho=-0.057$) and minimum deviation coefficient ($\rho=-0.130$), and the final log cell count.

\section{Inverse Coefficient of Variation and Coefficient of Variation of the Log Cell Area}

Recall image \ref{cho-15}, where one can observe that after a monoclone is seeded, how the sparse cells and cell colonies shift and vary as they colonise the well. Also, recall the confluence growth speed graph in figure \ref{c-speed-raw}, where one can observe that the confluence does not follow a smooth curve. Furthermore, recall Kolmogorov-Smirnov analysis from section \ref{log-normal}, where we demonstrate how the cell area distribution changed from normal to gamma and finally to log-normal as the cell population increases. Now, these reasons motivate us to find to model how the cell area distribution varies over time. Two such measures of distributions are the coefficient of variation, i.e. standard deviation divided by the mean, and the inverse coefficient of variation, i.e. mean divided by the standard deviation. Thus, following similar reasonings as in section (\ref{cell-count}), to model the inverse coefficient of variation, i.e. $y_+(x) = \frac{ \mathbb{E}(\log(\text{cell area}))}{\mathbb{S}\text{td}(\log(\text{cell area}))} $, and the coefficient of variation, i.e. $y_-(x) = \frac{\mathbb{S}\text{td}(\log(\text{cell area}))} { \mathbb{E}(\log(\text{cell area}))}$, of  the log cell area as a function of time, we arrive at the following theorem,

\begin{Theorem}[Homogeneity ($+$ve)/ Variation ($-$ve) Growth Model]\label{cell-homogeneity}
The inverse coefficient of variation and coefficient of variation of  the log cell area growth in a confined space can be modelled by the following equation,
\begin{align}
y_\pm(x) =  \beta_\pm \left[\alpha_\pm + \frac{1}{2} \pm  \frac{ 3 \sqrt{3}}{(\theta_\pm)^2}\boldsymbol{\sigma}^{(2)}(\theta_\pm (x - \gamma_\pm)) \right] ,\label{cell-homo}
\end{align}
where $y_+(x)$ is the inverse coefficient of variation of the log cell area, $y_-(x)$ is the coefficient of variation of the log cell area, $x$ is the time, $\boldsymbol{\sigma}^{(2)}(\cdot)$ is the second order derivative of the sigmoid function (definition \ref{sigmoid-diff}),
\begin{align*}
\theta_\pm  = \frac{8}{9} \sqrt{3} \frac{\omega_\pm}{\beta_\pm},
\end{align*}
$\beta_+$ is the homogeneity growth capacity, $\beta_-$ is the variation growth capacity, $\omega_+$ is the homogeneity growth rate, $\omega_-$ is the variation growth rate, $\gamma_+$ is the homogeneity growth time, $\gamma_-$ is the variation growth time, 
\begin{align*}
\alpha_\pm = \mathbb{E}\left[ \frac{1}{\beta_\pm} y_\pm(x) \mp  \frac{ 3 \sqrt{3}}{(\theta_\pm)^2}\boldsymbol{\sigma}^{(2)}(\theta_\pm (x - \gamma_\pm)) \mid (\beta_\pm, \gamma_\pm, \omega_\pm) \right] - \frac{1}{2} ,
\end{align*}
is the minimum homogeneity ($+$ve) / variation ($-$ve) coefficient, and where $\beta_\pm$, $\omega_\pm$ and $\gamma_\pm$ are the only independent variables of the model, and $x$ is the only independent variable of the dataset.
\end{Theorem}

Interpretation of theorem \ref{cell-homogeneity} is is rather statistical and abstract. However, the most intuitive interpretation is that homogeneity explains how closely the cell areas distributed from its mean value (i.e. how uniform the cells are), while variation describes how further away the cell areas distributed away from its mean value (i.e. how varied the cells are), and  equation (\ref{cell-homo}) shows how these values changed over time as the cell population grows. With some numerical analysis, we can show that $\gamma_+ < \gamma_-$. Thus, from the analysis that we conducted in section \ref{log-normal},  we may surmise that the time intervals $(0, \gamma_+)$, $(\gamma_+, \gamma_-)$ and $(\gamma_-, \infty)$ related to three distinct distributions that the cell area may take during its growth cycle. Now, taking first and second order derivatives of the equation (\ref{cell-homo}), we arrive at the following corollary,

\begin{corollary}
Theorem \ref{cell-homogeneity} implies that the homogeneity ($+$ve) / variation ($-$ve) growth speed and homogeneity ($+$ve) / variation ($-$ve)  growth acceleration can respectively expressed as follows,
\begin{align*}
\frac{d y_\pm(x)}{dx}  &= \pm \omega_\pm \left[ \frac{8}{(\theta_\pm)^3} \boldsymbol{\sigma}^{(3)}(\theta_\pm (x - \gamma_\pm)) \right]~\text{and}  \\
\frac{d^2 y_\pm(x)}{dx^2} &= \pm \omega_\pm \theta_\pm\left[ \frac{8}{(\theta_\pm)^4} \boldsymbol{\sigma}^{(4)}(\theta_\pm (x - \gamma_\pm))  \right],
\end{align*}
where $\boldsymbol{\sigma}^{(3)}(\cdot)$ and  $\boldsymbol{\sigma}^{(4)}(\cdot)$ are the third and fourth order derivatives of the sigmoid function (definition \ref{sigmoid-diff}).
\end{corollary}

\subsection{Numerical Modelling}

Given a dataset $(x, y_\pm)$, to find a $(\beta_\pm, \gamma_\pm, \omega_\pm, \alpha_\pm)$-set, we precent the following algorithm.\\

\emph{Step 1:} Normalise the dataset. First normalise the time data points as $x_{\text{norm}} =\frac{x}{ x_{\text{max}} }$, where $x_{\text{max}} =\max(x)$. Then, normalise the inverse coefficient of variation and coefficient of variation of the log area as $(y_\pm)_{\text{norm}} =\frac{y_\pm - (y_\pm)_{\text{min}}}{ (y_\pm)_{\text{max}} }$, where $(y_\pm)_{\text{min}} = \mathbb{E}(y_\pm) -  \mathbb{S}\text{td}(y_\pm) $ and $(y_\pm)_{\text{max}} =  2 \mathbb{S}\text{td}(y_\pm)$.\\

\emph{Step 2:} Find lower and upper bounds for the parameters. Using  \emph{SciPy} \emph{curve\_fit} function with \emph{maxfev}$=10,000$  \cite{scipy.curvefit}, fit the normalised data to the following equation,
\begin{align*}
(y_\pm)_{\text{norm}} = \pm \frac{81}{64} \sqrt{3}  \frac{ b_\pm}{( d_\pm)^2} \boldsymbol{\sigma}^{(2)}\left(\frac{8}{9} \sqrt{3}  d_\pm (x_{\text{norm}} - c_\pm) \right),
\end{align*}
where $b_\pm$, $c_\pm$ and $d_\pm$ bounded below by $0$ and above by $2$.\\

\emph{Step 3:} Find normalised homogeneity/variation growth parameters. Using the bounds $0<\beta_\pm^0 <b_\pm$, $c_\pm< \gamma_\pm^0 < 2$ and $0 < \omega^0_\pm< b_\pm d_\pm$, and using \emph{SciPy} \emph{curve\_fit} function  with \emph{maxfev}$=10,000$  \cite{scipy.curvefit}, fit the data to the following equation,
\begin{align*}
(y_\pm)_{\text{norm}} = \pm 3 \sqrt{3} \frac{\beta^0_\pm}{(\theta^0_\pm)^2} \big[\boldsymbol{\sigma}^{(2)}\big(\theta^0_\pm(x_{\text{norm}} - \gamma^0_\pm)\big) - \boldsymbol{\sigma}^{(2)}\big(-\theta^0_\pm \gamma^0_\pm\big) \big],
\end{align*}
where $\theta^0_\pm = \frac{8}{9} \sqrt{3} \frac{\omega^0_\pm}{\beta^0_\pm}$. With $\beta_\pm^0$, $ \gamma_\pm^0$ and $\omega^0_\pm$, find $\alpha_\pm^0$  as follows,
\begin{align*}
\alpha^0_\pm  =\mathbb{E}\left[ (y_\pm)_{\text{norm}} \mp 3 \sqrt{3} \frac{\beta^0_\pm}{(\theta^0_\pm)^2} [\boldsymbol{\sigma}^{(2)}\big(\theta_\pm^0 \big(x_{\text{norm}} - \gamma_\pm^0\big)\big) \big]\right] + \frac{(y_\pm)_{\text{min}}}{(y_\pm)_{\text{max}}} - \frac{1}{2} .
\end{align*}

\emph{Step 4:} Unnormalise the radial growth parameters as $\alpha_\pm = \frac{1}{\beta_\pm^0}  \alpha^0_\pm$, $\beta_\pm  = (y_\pm)_{\text{max}} \beta_\pm^0 $,  $\gamma_\pm  = x_{\text{max}} \gamma_\pm^0 $ and  $\omega_\pm  = \frac{(y_\pm)_{\text{max}} }{x_{\text{max}}} \omega_\pm^0 $.\\

For sample datasets, along with working algorithms, please see the links in the footnote\footnote{Homogeneity growth model \emph{Colab} notebook with a sample dataset,  a working algorithm and plots: 
\url{https://drive.google.com/file/d/185ugICaiNdHMnm4QytpWQo7UKA-TpMCa/view?usp=drive_link}\\
Variation growth model \emph{Colab} notebooks with a sample dataset,  a working algorithms and plots: 
\url{https://drive.google.com/file/d/1I9SvnXVH9qw93rtq_c_a3Yyj4LrrGi6y/view?usp=drive_link}} for a \emph{Colab} notebooks.

\subsection{Experimental Results}

In this section, we fit the our homogeneity  growth model and variation growth model  (theorem \ref{cell-homogeneity}) to \emph{CHO2023} dataset.\\

\begin{figure}[h!]
     \centering
     \begin{subfigure}[b]{0.3\textwidth}
         \centering
         \includegraphics[height=\textwidth]{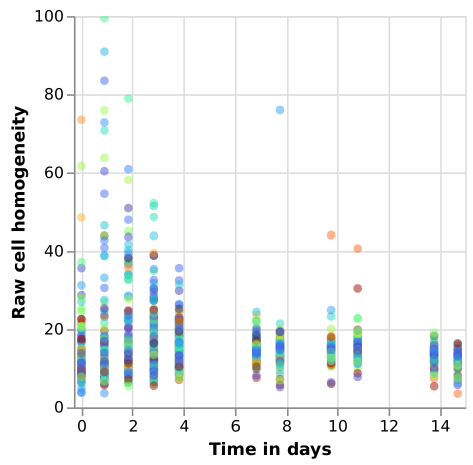}
         \caption{Raw log area ICV.}
         \label{h-capacity-raw}
     \end{subfigure}
     \hfill
     \begin{subfigure}[b]{0.3\textwidth}
         \centering
         \includegraphics[height=\textwidth]{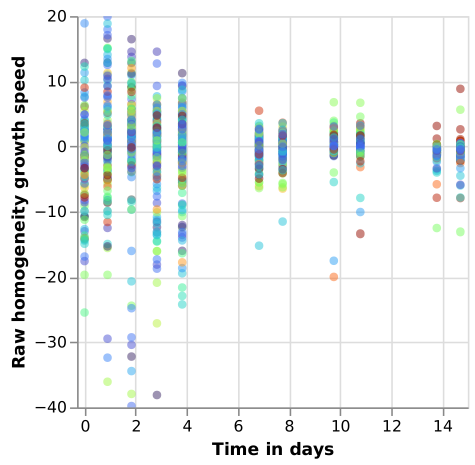}
         \caption{Raw speed.}
     \end{subfigure}
     \hfill
     \begin{subfigure}[b]{0.3\textwidth}
         \centering
         \includegraphics[height=\textwidth]{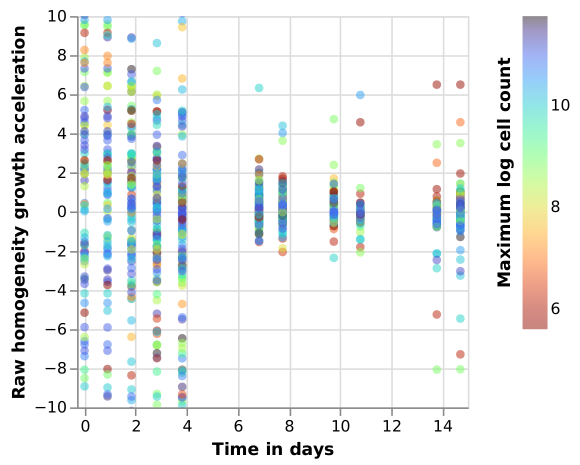}
         \caption{Raw acceleration.}
     \end{subfigure}
	\hfill
     \begin{subfigure}[b]{0.09\textwidth}
     \end{subfigure}
     \hfill
     \begin{subfigure}[b]{0.3\textwidth}
         \centering
         \includegraphics[height=\textwidth]{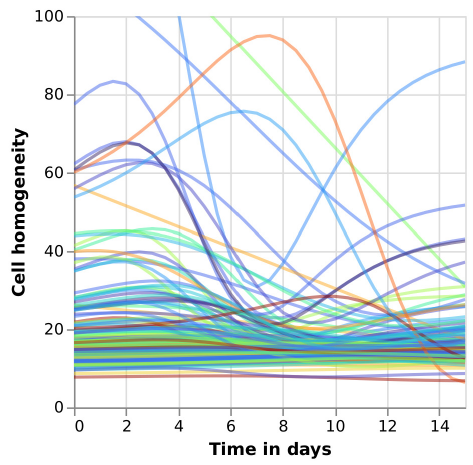}
         \caption{Log area ICV.}
         \label{h-capacity}
     \end{subfigure}
     \hfill
     \begin{subfigure}[b]{0.3\textwidth}
         \centering
         \includegraphics[height=\textwidth]{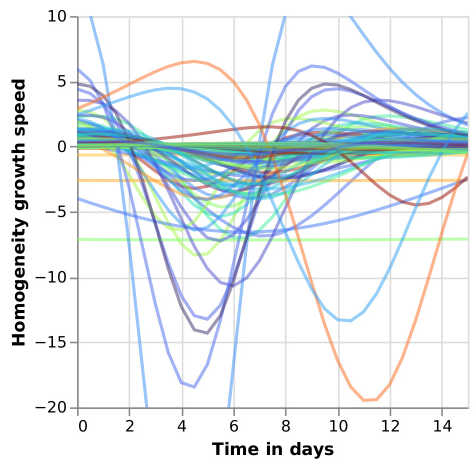}
         \caption{Speed.}
     \end{subfigure}
     \hfill
     \begin{subfigure}[b]{0.3\textwidth}
         \centering
         \includegraphics[height=\textwidth]{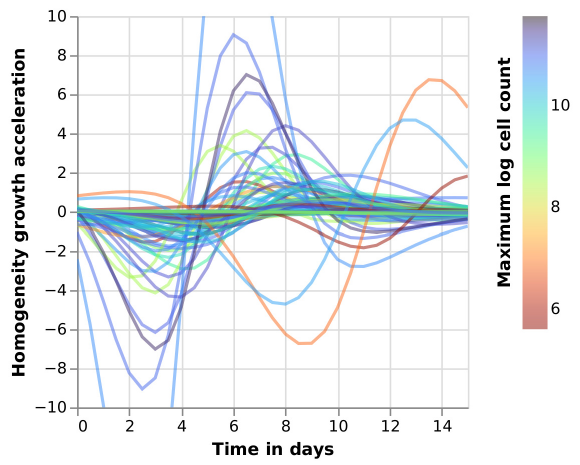}
         \caption{Acceleration.}
     \end{subfigure}
	\hfill
     \begin{subfigure}[b]{0.09\textwidth}
     \end{subfigure}
        \caption{Homogeneity growth patterns: raw vs modelled.}
        \label{experiment-h}
\end{figure}

Figure \ref{experiment-h} shows inverse coefficient of variation of the log cell area, homogeneity growth speed and homogeneity growth acceleration, along with what is implied by our model. Sub-figures \ref{h-capacity-raw} and \ref{h-capacity} show that homogeneity increases from its baseline to $(\alpha_+ +1) \beta_+$, declines below the baseline to $\alpha_+ \beta_+$ and increases back to the baseline $(\alpha_+ +\frac{1}{2})\beta_+$, as the population grows. Perhaps the peak, the trough and the baseline represent three different cell area distributions, and what we are observing is cell areas are moving from one distribution to another as the population grows.\\

\begin{figure}[h!]
     \centering
     \begin{subfigure}[b]{0.23\textwidth}
         \centering
         \includegraphics[height=\textwidth]{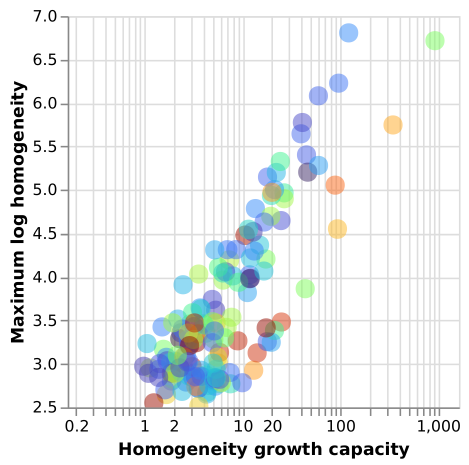}
         \caption{$\beta_+$}
     \end{subfigure}
     \hfill
     \begin{subfigure}[b]{0.23\textwidth}
         \centering
         \includegraphics[height=\textwidth]{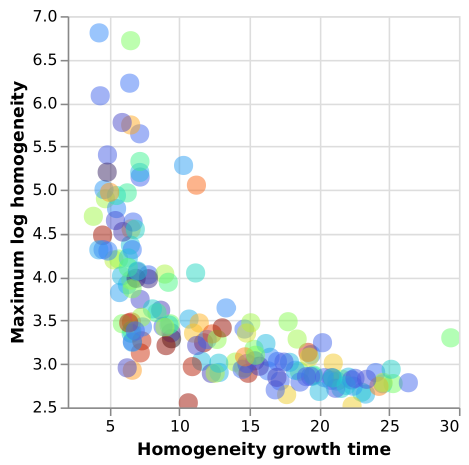}
         \caption{$\gamma_+$}
     \end{subfigure}
     \hfill
     \begin{subfigure}[b]{0.23\textwidth}
         \centering
         \includegraphics[height=\textwidth]{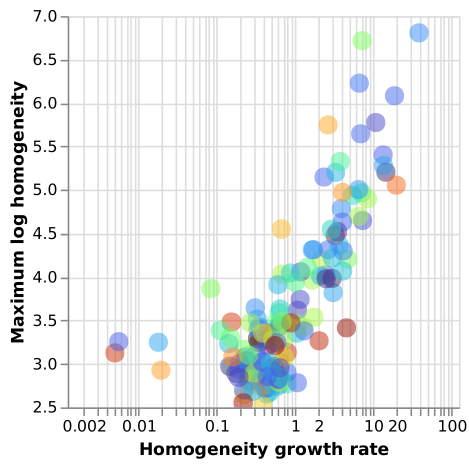}
         \caption{$\omega_+$}
     \end{subfigure}
	\hfill
     \begin{subfigure}[b]{0.23\textwidth}
         \centering
         \includegraphics[height=\textwidth]{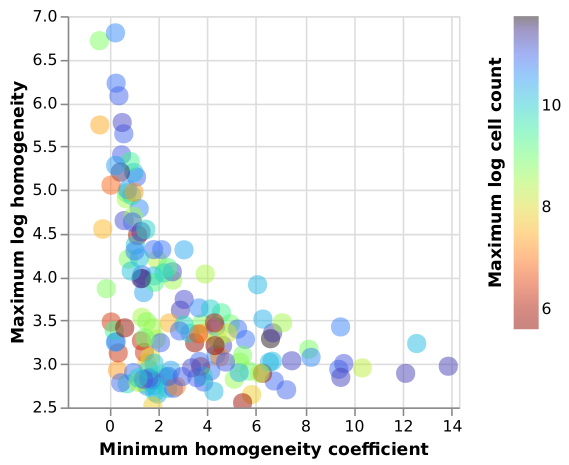}
         \caption{$\alpha_+$}
     \end{subfigure}
	\hfill
     \begin{subfigure}[b]{0.09\textwidth}
     \end{subfigure}
        \caption{Homogeneity growth variables.}
        \label{h-var}
\end{figure}

Figure \ref{h-var} shows the homogeneity growth variables with respect to maximum of the log inverse coefficient of variation of the log cell area. We are unable to find any meaningful relationship with homogeneity growth variables and the final log cell count (i.e. homogeneity growth capacity: $\rho = -0.070$, homogeneity growth time: $\rho = -0.049$, homogeneity growth rate: $\rho = 0.088$ and minimum homogeneity coefficient: $\rho = 0.113$).\\

\begin{figure}[h!]
     \centering
     \begin{subfigure}[b]{0.3\textwidth}
         \centering
         \includegraphics[height=\textwidth]{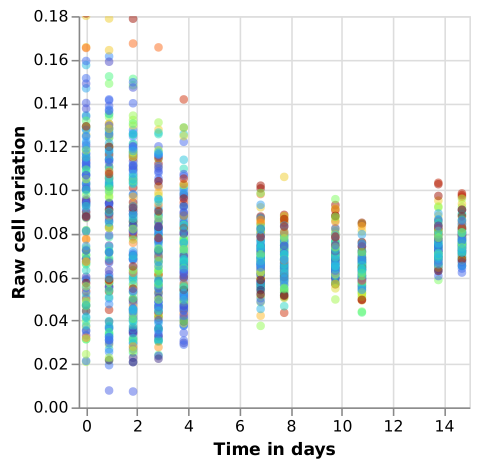}
         \caption{Raw log area CV.}
         \label{v-capacity-raw}
     \end{subfigure}
     \hfill
     \begin{subfigure}[b]{0.3\textwidth}
         \centering
         \includegraphics[height=\textwidth]{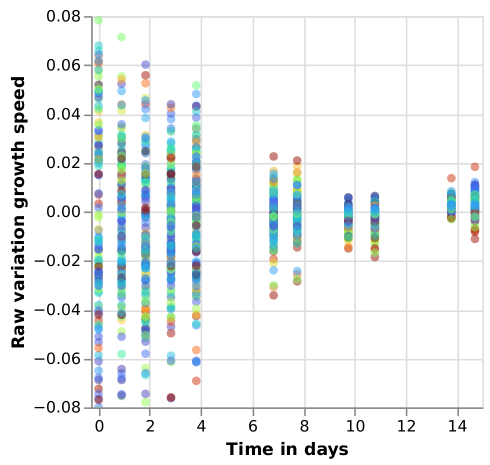}
         \caption{Raw speed.}
     \end{subfigure}
     \hfill
     \begin{subfigure}[b]{0.3\textwidth}
         \centering
         \includegraphics[height=\textwidth]{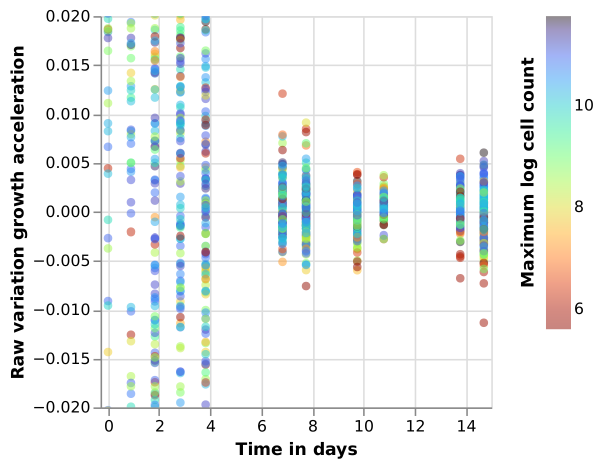}
         \caption{Raw acceleration.}
     \end{subfigure}
	\hfill
     \begin{subfigure}[b]{0.09\textwidth}
     \end{subfigure}
     \hfill
     \begin{subfigure}[b]{0.3\textwidth}
         \centering
         \includegraphics[height=\textwidth]{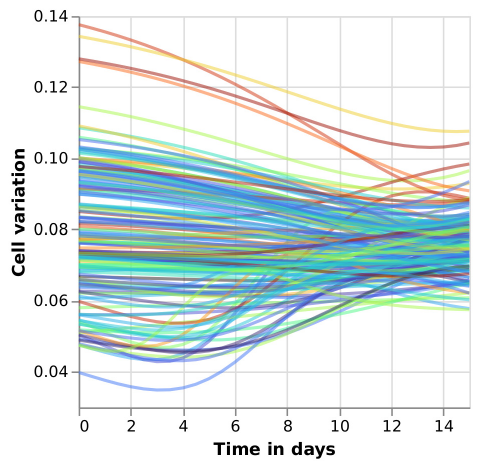}
         \caption{Log area CV.}
         \label{v-capacity}
     \end{subfigure}
     \hfill
     \begin{subfigure}[b]{0.3\textwidth}
         \centering
         \includegraphics[height=\textwidth]{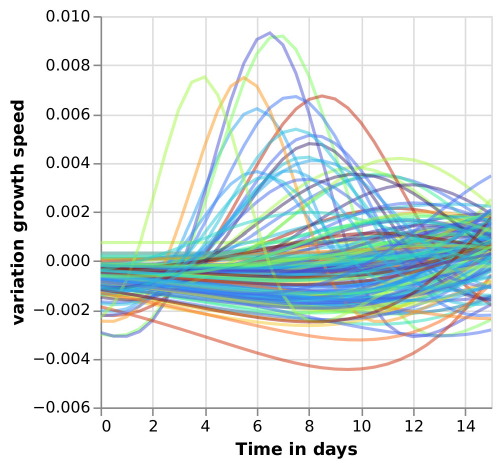}
         \caption{Speed.}
     \end{subfigure}
     \hfill
     \begin{subfigure}[b]{0.3\textwidth}
         \centering
         \includegraphics[height=\textwidth]{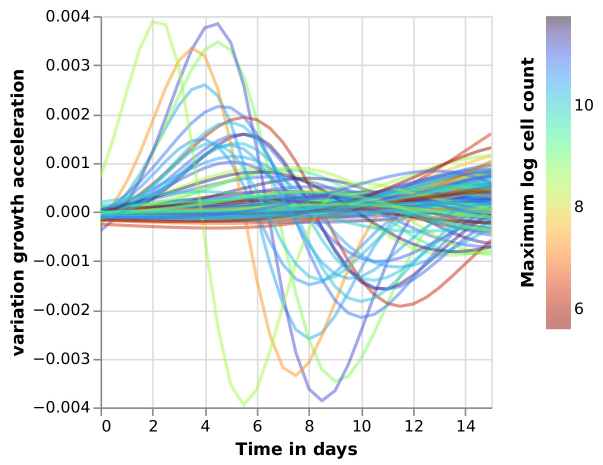}
         \caption{Acceleration.}
     \end{subfigure}
	\hfill
     \begin{subfigure}[b]{0.09\textwidth}
     \end{subfigure}
        \caption{Variation  growth  patterns: raw vs modelled.}
        \label{experiment-v}
\end{figure}

Figure \ref{experiment-v} shows coefficient of variation of the log cell area, variation growth speed and variation growth acceleration, along with what is implied by our model. Sub-figures \ref{v-capacity-raw} and \ref{v-capacity} show that variation  declines from its baseline to  $\alpha_-\beta_-$,  increases  above the baseline to $(\alpha_-+1)\beta_-$ and decreases back to the baseline $(\alpha_-+\frac{1}{2})\beta_-$, as the population grows. We also observe that for $99.2\%$ of the samples, we have $\gamma_+ < \gamma_-$, where, on average, $\gamma_-$ exceeds $\gamma_+$ by $3.89$ days, giving us three distinct time intervals, i.e. before $\gamma_+$, between $\gamma_+$ and $\gamma_-$, and after $\gamma_-$. Again, perhaps we are observing the temporal boundaries of  the cell areas distributions as it moves from one distribution to another as the population grows.\\

\begin{figure}[h!]
     \centering
     \begin{subfigure}[b]{0.23\textwidth}
         \centering
         \includegraphics[height=\textwidth]{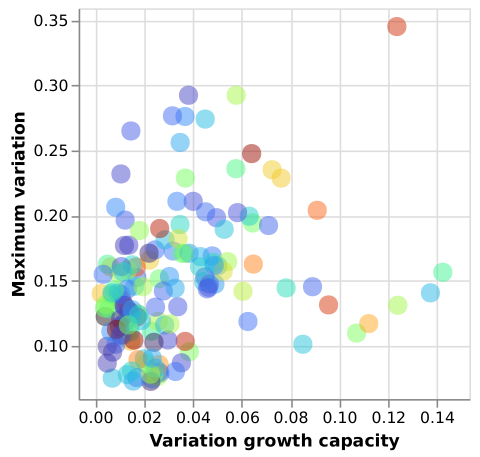}
         \caption{$\beta_-$}
     \end{subfigure}
     \hfill
     \begin{subfigure}[b]{0.23\textwidth}
         \centering
         \includegraphics[height=\textwidth]{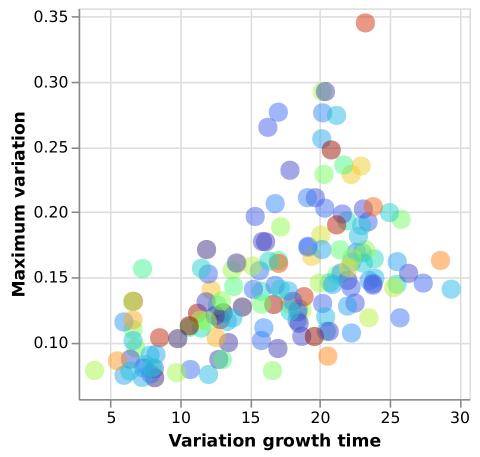}
         \caption{$\gamma_-$}
     \end{subfigure}
     \hfill
     \begin{subfigure}[b]{0.23\textwidth}
         \centering
         \includegraphics[height=\textwidth]{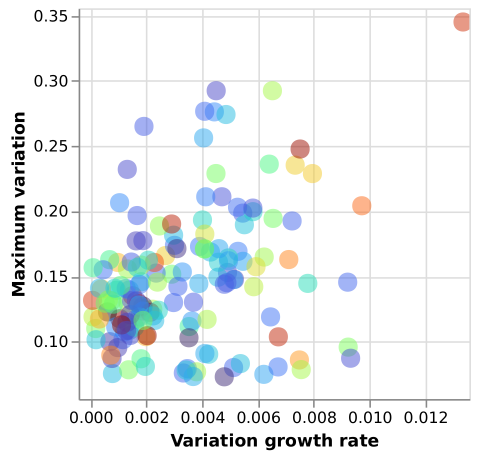}
         \caption{$\omega_-$}
     \end{subfigure}
	\hfill
     \begin{subfigure}[b]{0.23\textwidth}
         \centering
         \includegraphics[height=\textwidth]{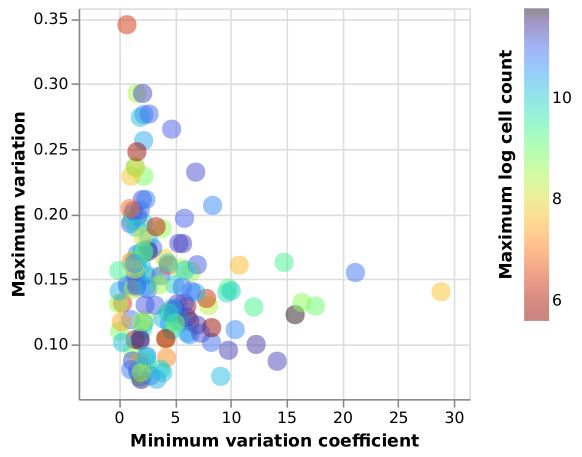}
         \caption{$\alpha_-$}
     \end{subfigure}
	\hfill
     \begin{subfigure}[b]{0.07\textwidth}
     \end{subfigure}
        \caption{Variation growth variables.}
        \label{v-var}
\end{figure}

Figure \ref{v-var} shows the variation growth variables with respect to maximum of the coefficient of variation of the log cell area. We are unable to find any meaningful relationship with variation growth variables and the final log cell count (i.e. variation growth capacity: $\rho = -0.180$, variation growth time: $\rho = 0.049$, variation growth rate: $\rho = -0.080$ and minimum variation coefficient: $\rho = 0.020$).

\section{Applications in Titer Measurements and Viability Predictions \label{sec-titer}}

Based on growth patterns of another 48 CHO monoclones observed over a 16 day period, we  calculate growth variables, and their titer and viability measurements, where titer is measured with Solentim ICON \cite{ICON} and viability is measured with Vi-CELL XR \cite{vicell}. Note that in this dataset, cells are seeded by hand, and thus, the probability of monoclonality is relatively low. Table \ref{tab} shows the correlation coefficients between the calculated growth variables, titer (i.e. total titer), titer per cell (i.e. total titer divided by the final cell count), daily titer production per cell (i.e. total titer divided by the integral of the cell count over time, where the integral is calculated with equation (\ref{cell-growth-model}) and \emph{SciPy} \emph{quad} function \cite{scipy.quad}) and viability, where cell count for the per cell basis is calculated with DeepInsight\textsuperscript{\textregistered} cell analysis software \cite{tekor}, and  boldface and underlined text show the most notable correlations per growth variable.\\

\begin{table}[h!]
\begin{tabular}{ |l||r|r|r|r|}
\toprule
\multirow{1}{*}{Growth Variable} & \multicolumn{4}{ |c |}{ Correlation coefficient} \\
\cmidrule{2-5}
 & \multicolumn{1}{ |c |}{ Titer} & \multicolumn{1}{ |c |}{ Titer}& \multicolumn{1}{ |c |}{ Titer} & \multicolumn{1}{ |c |}{ Via-}\\
 &  \multicolumn{1}{ |c |}{ (total)} &  \multicolumn{1}{ |c |}{ per cell\! \! } & \multicolumn{1}{|c |}{ per day\!\! } & \multicolumn{1}{ |c |}{ bility}\\
 &  &   &  \multicolumn{1}{ |c |}{  per cell\! \! }   & \\
\midrule
$\beta_n$  Population growth capacity	&	\textbf{0.275}	&	-0.105	&	-0.023	&	0.009	\\
$\gamma_n$  Population growth time	&	0.244	&	-0.046	&	\textbf{0.472}	&	0.010	\\
$\omega_n$  Population growth rate	&	-0.155	&	-0.114	&	\underline{\textbf{-0.538}}	&	0.071	\\
$\alpha_n$  Minimum population coefficient 	&	\underline{\textbf{-0.300}}	&	-0.096	&	-0.129	&	-0.011	\\
$\epsilon_n$  Incipient population growth capacity \!\!\!  \!\!\!  \!\!\!  \!\!\! &	0.055	&	0.228	&	\textbf{0.288}	&	-0.004	\\
$\rho_n$  Incipient population growth rate	&	-0.057	&	-0.231	&	\textbf{-0.291}	&	0.007	\\
\hline
$\beta_c$  Confluence growth capacity	&	\textbf{0.272}	&	-0.082	&	-0.092	&	-0.021	\\
$\gamma_c$  Confluence growth time	&	0.247	&	-0.017	&	\textbf{0.465}	&	-0.089	\\
$\omega_c$  Confluence growth rate	&	-0.143	&	-0.059	&	\textbf{-0.481}	&	0.107	\\
$\alpha_c$  Minimum confluence coefficient	&	\textbf{-0.287}	&	-0.0137	&	-0.038	&	-0.059	\\
\hline
$\beta_r$  Radial growth capacity	&	-0.165	&	\textbf{-0.167}	&	-0.150	&	-0.037	\\
$\gamma_r$  Radial growth time	&	0.182	&	0.0057	&	\textbf{0.245}	&	0.218	\\
$\omega_r$  Radial growth rate	&	-0.139	&	-0.117	&	\textbf{-0.145}	&	-0.011	\\
$\alpha_r$  Minimum radial coefficient 	&	-0.038	&	-0.067	&	-0.029	&	\textbf{0.081}	\\
$\epsilon_r$  Incipient radial growth capacity	&	\textbf{0.269}	&	-0.004	&	0.169	&	0.100	\\
$\rho_r$  Incipient radial growth rate	&	\textbf{-0.269}	&	-0.007	&	-0.169	&	0.104	\\
\hline
$\beta_d$  Deviation growth capacity	&	0.144	&	\underline{0.242}	&	0.130	&	\textbf{-0.305}	\\
$\gamma_d$  Deviation growth time	&	0.190	&	0.120	&	\textbf{0.251}	&	0.041	\\
$\omega_d$  Deviation growth rate	&	0.096	&	0.178	&	0.025	&	\underline{\textbf{-0.338}}	\\
$\alpha_d$  Minimum deviation coefficient	&	\textbf{0.099}	&	0.009	&	0.033	&	0.093	\\
\hline
$\beta_+$  Homogeneity growth capacity	&	-0.175	&	-0.136	&	\textbf{-0.186}	&	-0.014	\\
$\gamma_+$  Homogeneity growth time	&	0.219	&	0.185	&	\textbf{0.302}	&	-0.117	\\
$\omega_+$  Homogeneity growth rate	&	\textbf{-0.200}	&	-0.159	&	-0.194	&	0.003	\\
$\alpha_+$  Minimum homogeneity coefficient	&	0.180	&	0.036	&	\textbf{0.302}	&	0.191	\\
\hline
$\beta_-$  Variation growth capacity	&	-0.223	&	-0.174	&	\textbf{-0.235}	&	0.014	\\
$\gamma_-$  Variation growth time	&	-0.014	&	-0.031	&	0.0143	&	\textbf{0.150}	\\
$\omega_-$  Variation growth rate	&	\textbf{-0.248}	&	-0.187	&	-0.2454	&	0.020	\\
$\alpha_-$  Minimum variation coefficient	&	-0.038	&	-0.030	&	\textbf{0.258}	&	0.122	\\

\bottomrule
\end{tabular}
\caption{Correlation coefficients of growth variables with respect to titer, titer per cell, titer per day per cell (i.e. titer rate per cell) and viability measurements.}
\label{tab}
\end{table}

Table \ref{tab} shows that total titer measurements are weakly positively-correlat- ed with variables that are associated with greater population growth, greater confluence growth, greater likelihood of monoclonality and greater capacity for the incipient cell to expand in size at a slower rate. It also shows no significant correlation between titer measurements per cell and growth variables. It further shows that rate of titer production per cell measurements are moderately positively-correlated with variables that are associated with slow cell growth, slow confluence growth and cells associated with longer periods of more homogenous cell area distributions. Finally, it shows that viability measurements are weakly positively-correlated with variables that are associated with low cell area standard deviations.\\

We find a weak negative-correlation between rate of titer production per cell and viability measurements ($\rho = -0.213$). Should one examine table \ref{tab} and recall table \ref{tab-inter}, one finds that the rate of titer production per cell is negatively correlated the population growth rate, where there is a negative interdependency between the population growth rate and the deviation growth rate, and where the the deviation growth rate is negatively correlated with viability measurements. This implies a interdependency between rate of  titer production per cell and viability measurements.  Recall that equation (\ref{inter}) shows that the confluence, the cell count, the mean cell area and the standard deviation of the cell area are all interdependent. Also, table \ref{tab} shows that total titer, titer production rate per cell and viability measurements are correlated with the confluence, the cell count, the mean cell area and the standard deviation growth variables. This results in the following theorem,

\begin{Theorem}
The interdependency of growth variables (see equation (\ref{inter})), and the correlation of titer and viability measurements with growth variables (see table table \ref{tab}) imply that the productivity and the health of a cell (also the overall population) are interdependent.
\end{Theorem}

Above analysis shows that growth variables can be used to predict the productivity and health of a cell; thus, justifying the efficacy of our numerical models. We confirm that these growth variables are indeed being used in predicting the early cell screening and cell selection. However, we cannot disclose the exact method that we use in such applications.\\

Note that growth capacities have the log of the dimensions of the property that we model, growth times have the dimensions of time and the growth rates have the dimensions of growth capacity divided by growth time. Thus, properties of all the dimensions and their ratios are independently covered by growth capacities, rates and times. This is the reason why we insist upon normalising $\alpha$s to make them dimensionally independent. This gives us another independent variable to correlate against a different property and to extract useful relations; thus, justifying the definitions of $\alpha$s in theorems \ref{full-cell-growth-model}, \ref{cell-confluence}, \ref{cell-area}, \ref{cell-d-thrm}, and \ref{cell-homogeneity}.

\section{Conclusions \label{con}}

In this article, we derived equations for modelling population growth, confluence growth, mean cell area growth, standard deviation of cell area growth, inverse coefficient of variation of cell area growth and coefficient of variation of cell area growth in a confined space. We achieved this by combining the sigmoid function (to model incipient population growth) and exponential of the sigmoid function (to model mature population growth), and using exponential of the sigmoid function and its derivatives (to model confluence, radial, deviation, homogeneity and variation growth). We also presented ways to numerically model our equations, along with algorithms, python scripts and sample datasets.\\

Based on growth patterns of 166 CHO monoclones observed over a 15 day period, and with the DeepInsight\textsuperscript{\textregistered} cell analysis software \cite{tekor}, we show that our population growth equation can capture the complex behaviour of population growth speed and the population growth acceleration with a significant degree of accuracy, including the shock of being seeded (predicted by the incipient growth capacity coefficient), the recovery after the shock (predicted by the incipient growth rate coefficient) and the growth until the area constraint of the well becomes a limiting factor. Our analysis shows that population growth capacity is a good predictor of how well a cell can colonise its environment. It is also a more accurate estimation of the cell count when the population starts to form vertical colonies. We also find the behaviour of confluence growth model is highly correlated with population growth model, as it also follows the exponential of the sigmoid curve.\\

Radial growth model implies that incipient population more likely to expand in size (area per cell wise) than to multiply, where this result is consistent with the reduction in the population growth  speed of the incipient population. It also implies that  the mean area of a cell tends to be large when it is sparse, and the area per cell decreases as cells multiply and start to form dense colonies. We further find that large sparse cells and cells that have the capacity to form very dense colonies has the potential to grow in to large populations. Deviation growth model implies that the standard deviation of the cell area is high for the incipient population, low for rapidly growing sparse cells, and high for dense cell colonies. Homogeneity (inverse coefficient of variation) and variation (coefficient of variation) growth models imply that there exists at least three distinct time intervals where the cell areas fall in to distinct probability distributions.\\

Based on the Kolmogorov-Smirnov analysis conducted on the area of the CHO monoclones, we find that the area per cell of the incipient population is normally distributed, the sparse cell population is gamma distributed and the dense colony population is log-normally distributed. As far as we are aware, this is the first time such goodness of fit for distribution is conducted on cell area data.\\

Based on growth patterns of further 48 CHO monoclones observed over a 16 day period, and taking their titer and viability measurements, we find that if a population has a high capacity to grow  (i.e. high population growth capacity and high confluence growth capacity), a greater likelihood of originating from a single progenitor (i.e. low minimum population coefficient and low minimum confluence coefficient) and a high capacity for the incipient cell to expand in size at a slower rate (i.e. high incipient radial growth capacity coefficient and low incipient radial growth rate coefficient), then the population likely to produce more titer. We also find that if a cell that takes a long time to multiply (i.e. high population growth time, high confluence growth time, low population growth rate and low confluence growth rate) and likely to have a very homogenous cell area distribution for a long period of time (i.e. high homogeneity growth time and high minimum homogeneity coefficient), then the cell likely to produce more titer per day. Finally, we find that if the population's area per cell has a low standard deviation (i.e. low deviation growth rate and low deviation growth capacity), then the population likely to be more viable. To put it simply, large populations that are less likely to form colonies results in greater titer measurements, slow growing populations with uniform areas per cell likely to results in greater titer production rate, and populations with areas per cell that are more uniform results in a more viable population.\\

Finally, with rigorous mathematical and statistical analysis, we demonstrated that the productivity and the health of a cell (also the overall population) are interdependent. For example, we find a negative correlation between rate of titer production per cell and viability measurements. Our mathematical analysis shows that the biological processes that under pin this result is the negative interdependency between the population growth rate and the deviation growth rate.\\

Our observations of positive correlations between cell size, cell growth and cell productivity appear to be consistent with the findings of Marshall \cite{marshall2012organelle}, our observations of cell areas can fall in to diffrent probability distribution during its growth cycle appear to be consistent with the findings of Jia \emph{et al.} \cite{jiacell}, our observation of rapidly growing cell areas being gamma distributed appear to consistent with the findings of Jia \emph{et al.} \cite{jiacell} and Golubev \cite{golubev2016applications}, and our observations of colony cell areas being log-normally distributed appears to be consistent with the findings of Lenz \emph{et al.} \cite{lenz2016estimating}. Also, any conclusions we my draw from titer and viability measurement are based on the assumption that titer measurements of Solentim ICON \cite{ICON} and viability measurements of  Vi-CELL XR \cite{vicell} are accurate. As a concluding remark, we note that our analysis is made possible due the accuracy of DeepInsight\textsuperscript{\textregistered} cell analysis software \cite{tekor}. We present our research findings in the hope that it may aid further discoveries from other researchers.

\section*{Acknowledgments}
We thank Wheeler Bio, Inc. \cite{wheelerbio} for providing the data of CHO cell scans, titer measurements and viability measurements.

\bibliographystyle{./model1-num-names}
\bibliography{bib} %
\biboptions{sort&compress}

\end{document}